\documentclass[10pt, conference, letterpaper]{IEEEtran}

\usepackage{etex}
\usepackage{mdwlist}

\usepackage{algpseudocode}
\usepackage{algorithm}
\usepackage{amsmath}
\usepackage{gensymb}
\usepackage{amsthm}
\usepackage{mathtools}
\usepackage{relsize}

\newcommand{\comment}[1]{}

\algnewcommand\algorithmicto{\textbf{to}}
\algnewcommand\RETURN{\State \textbf{return} }

\usepackage{bm}
\usepackage{amssymb}

\usepackage{varwidth}

\DeclareMathOperator*{\argmin}{arg\,min}

\usepackage{enumerate}
\usepackage{amssymb}
\usepackage{multirow}
\usepackage{dsfont}
\usepackage{citesort}
\usepackage{array}
\newtheorem{definition}{Definition}
\newtheorem{lemma}{Lemma}
\newtheorem{theorem}{Theorem}

\usepackage{listings}
\usepackage[utf8x]{inputenc} 
\usepackage{graphicx}
\usepackage{float}
\usepackage{epstopdf}
\usepackage{bbm}
\usepackage{pbox}
\usepackage{color}
\usepackage[section,subsection,subsubsection]{extraplaceins}
\usepackage{multirow}
\usepackage{subfig}
\usepackage{booktabs}

\newcommand{\RNum}[1]{\uppercase\expandafter{\romannumeral #1\relax}}
\newcommand{\suchthat}{\;\ifnum\currentgrouptype=16 \middle\fi|\;}

\usepackage{amsmath}
\usepackage{tikz}
\usetikzlibrary{automata,arrows,positioning,calc}

\usetikzlibrary{arrows}
\tikzstyle{int}=[draw, fill=blue!20, minimum size=2em]
\tikzstyle{init} = [pin edge={to-,thin,black}]

\usetikzlibrary{positioning,chains,fit,shapes,calc, arrows, shadows,trees}
\usepackage{calc}
\usepackage{multicol}
\usepackage{lipsum}
\usepackage{mathtools} 
\usepackage{ulem}
\usepackage{mathrsfs}

\title{Integrating Sub-6~GHz and Millimeter Wave to Combat Blockage: Delay-Optimal Scheduling}
{\author{\IEEEauthorblockN{Guidan Yao\IEEEauthorrefmark{1}, Morteza Hashemi\IEEEauthorrefmark{2}, and Ness B. Shroff \IEEEauthorrefmark{1}\IEEEauthorrefmark{3}}\IEEEauthorblockA{\IEEEauthorrefmark{1}  Department of Electrical and Computer Engineering, Ohio State University} 
\IEEEauthorblockA{\IEEEauthorrefmark{2}Department of Electrical Engineering and Computer Science, University of Kansas (KU)}
\IEEEauthorblockA{\IEEEauthorrefmark{3}Department of Computer Engineering and Computer Science, Ohio State University}}}

\begin{document}
\maketitle

\normalem

\begin{abstract}
Millimeter wave (mmWave) technologies have the potential to achieve very high data rates, but suffer from intermittent connectivity. In this paper, we provision an architecture to integrate
sub-6 GHz and mmWave technologies, where we
incorporate the sub-6 GHz interface as a fallback data transfer
mechanism to combat blockage and intermittent connectivity of the mmWave communications. To this end, we investigate
the problem of scheduling data packets across the mmWave and 
sub-6 GHz interfaces such that the average delay of system is
minimized. This problem can be formulated as Markov Decision Process. We first investigate
the problem of discounted delay minimization, and prove that the
optimal policy is of the threshold-type, i.e., \emph{data packets should
always be routed to the mmWave interface as long as the
number of packets in the system is smaller than a threshold}. Then, we show that the
results of the discounted delay problem hold for the average
delay problem as well. {Through numerical results, we demonstrate that under heavy traffic, integrating sub-6 GHz with mmWave can reduce the average delay by up to 70\%. Further, our scheduling policy substantially reduces the delay over the celebrated MaxWeight policy.}
\end{abstract}

\section{Introduction}
The annual amount of mobile data is projected to surpass 130 exabits by 2020 \cite{khan2011mmwave}. With such rapid increases in mobile data traffic, we are facing unprecedented challenges due to the shortage of wireless spectrum. To mitigate the problem of spectrum scarcity, the millimeter wave (mmWave) band, ranging from 30 GHz to 300 GHz,  provides a promising solution \cite{rappaport2013millimeter}. However, before mmWave communications can become a reality, there exist several significant challenges that need to be overcome.
In particular, mmWave channels can be highly variable with intermittent on-off periods. Due to small wavelengths in the mmWave band, most objects, such as concrete walls, a human body or even rain drops, may cause blocking and reflections as opposed to scattering and diffraction in the sub-6 GHz frequencies. In this case, blockage may completely break the mmWave link and result in an almost zero delivery rate \cite{sur2017wifi,niu2015blockage,genc2010robust}. In the provisioned applications of mmWave, human blockage is  one of the main challenges that can increase the path loss by more than 20 dB \cite{collonge2004influence,slezak2018empirical,sur201560,sato1998estimation}. 

To demonstrate the effect of human blockage on mmWave links, we have conducted a set of measurements with a stationary transmitter and a mobile receiver that moves away from the transmitter with the speed of 1 m/s. During the time intervals $200-300$ and $500-600$ ms, a human body blocks the line-of-sight (LOS) path between the transmitter and receiver. Figure~\ref{fig:system_setup} shows our basic experimental setup and Fig. \ref{blockageexperiment} depicts the strength of received signal at the mobile receiver over time \cite{hashemi2018out}. 
From the results, we see that the received signal strength falls to almost zero under blockage, which can be modeled as an OFF  or unavailable period. Therefore, the mmWave link exhibits an ON/OFF connectivity pattern under blockage scenarios such that during the OFF periods, delivery rate and delay performance can highly degrade.

\begin{figure}[t!]
\begin{center}
     \includegraphics[scale=.6,trim = 1.5cm 2.5cm 1.5cm 1.7cm, clip]{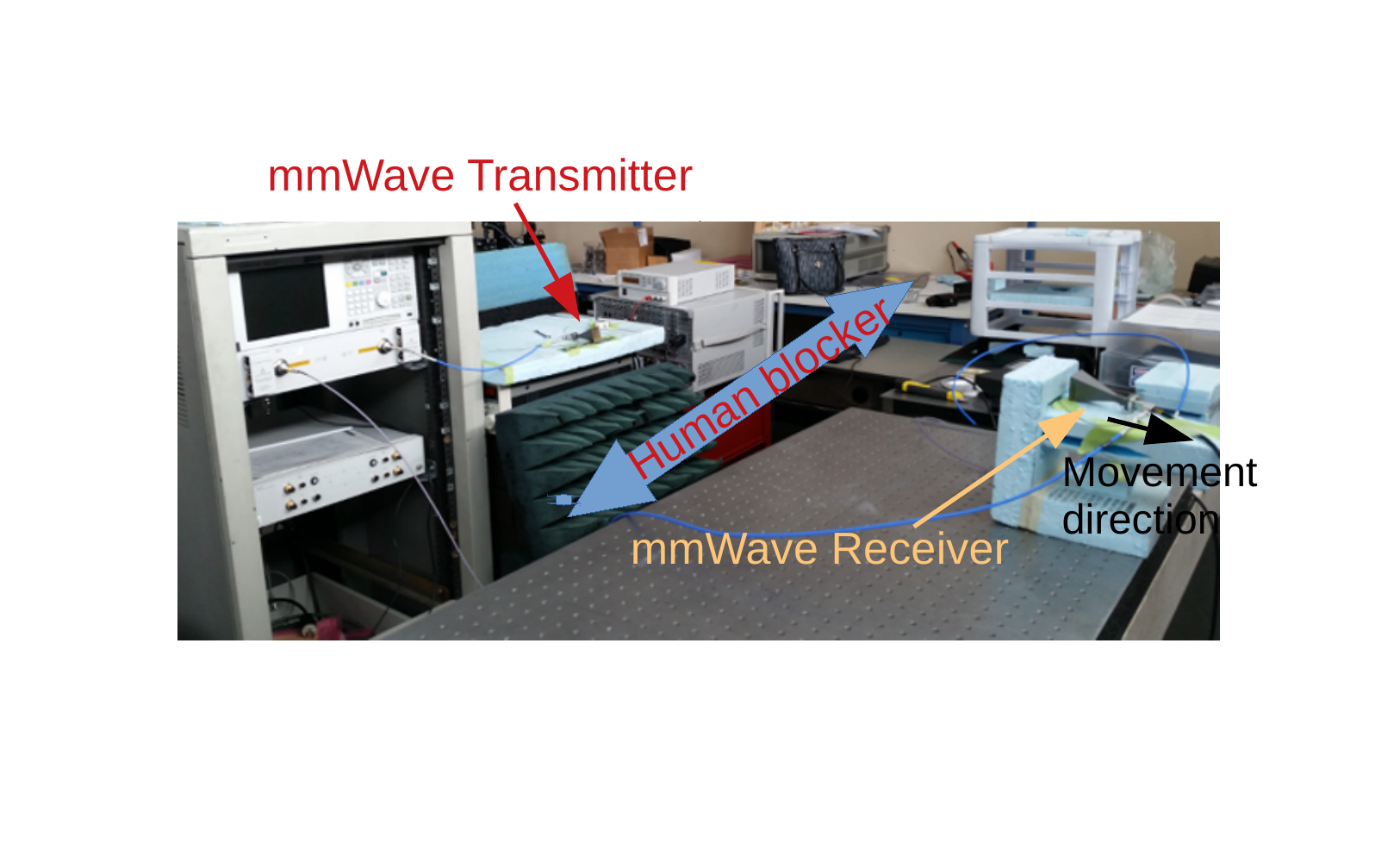} 
  \caption{Measurement setup and experiment scenario to investigate the effect of human blockage on mmWave channels. }
  \label{fig:system_setup}
  \end{center}
\end{figure}
\begin{figure}[t!]
\centering
\includegraphics[scale=.30]{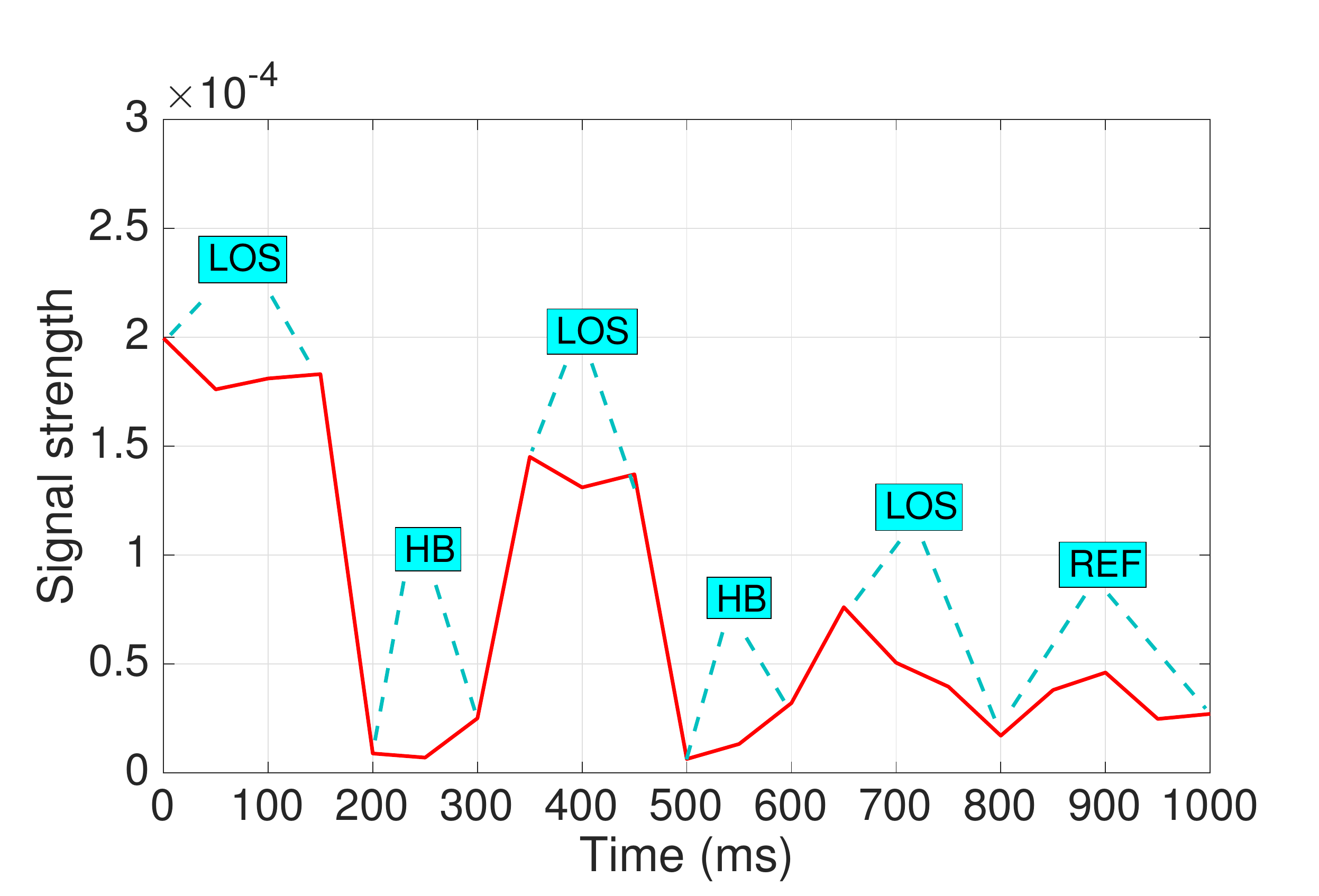}
\caption{Received mmWave signal strength under line of sight
(LOS), human blocker (HB), and reflection (REF) \cite{hashemi2018out}.}
\label{blockageexperiment}
\end{figure}

In order to mitigate the effects of intermittent connectivity, especially for delay-sensitive applications, several methods have been proposed. For instance, the authors in \cite{genc2010robust} and \cite{singh2007millimeter} exploit reflection paths and multi-hop paths to combat blockage.  
These methods are \emph{reactive} in the sense that the search for an alternative path is triggered after blockage occurs. However, since the link speed of the mmWave interface (multi-Gbps) is comparable to the speed at which a typical processor in a smart device operates, these methods may not be able to track and respond to channel variations in real-time. Therefore, it necessitates the use of a reasonably large buffer at the mmWave interface along with \emph{proactive solutions} to complement this design. In addition to the aforementioned methods, there exist several works on integrating the mmWave and sub-6 GHz technologies. For instance, due to spatial correlation, information of the sub-6 GHz channel can be extracted to reduce mmWave beamforming overhead \cite{nitsche2015steering,aliestimating}, while \cite{hashemi2018efficient} uses an online learning method to enhance beamforming delay.    
Moreover, the authors in \cite{hashemi2017energy,hashemi2018out,semiari2017joint} consider resource allocation and cooperative communication between the sub-6 GHz and mmWave to maximize either the throughput of the system or the quality-of-service per user application. 

Although an integrated mmWave/sub-6 GHz architecture has been previously proposed, the delay minimization problem in this integrated architecture has not been explored yet. In particular, 
we raise the following questions that: \emph{Would the system delay benefit from the sub-6 GHz interface? If so, when should the sub-6 interface be used so that the delay of the system is minimized?} 
In this paper, we exploit the sub-6 GHz interface as a fallback data transfer mechanism, and investigate the problem of delay-optimal scheduling across the sub-6 GHz and mmWave interfaces. We develop a proactive scheduling policy that is expressed in terms of the queue length of the mmWave and sub-6 GHz interfaces.

In order to obtain a delay-optimal policy, we first consider minimizing the \emph{expected total discounted delay}. We obtain three rules for the delay-optimal scheduling using value iteration in Markov Decision Process (MDP). 
Next, we collapse our system state space from four dimension to three dimension, and further \emph{demonstrate that the optimal policy for the  discounted delay problem is of a threshold-type.} Finally, we extend our results to the \emph{average delay problem}. Through simulations, we show that it is important to use the sub-6 GHz interface especially when the mmWave is unavailable with high probability and confirm that such a threshold-type policy improves the average delay performance while achieving similar throughput performance as the throughput-optimal and well-studied MaxWeight policy \cite{tassiulas1992stability}. 

In summary, our main contributions are as follows: (i) we formulate the discounted delay optimality problem in the integrated sub-6 GHz/mmWave architecture as an MDP and provide partial characteristics of the optimal policy. Based on the findings, we propose a threshold-type policy and then prove its optimality; (ii) we further show that the proposed policy is also optimal for the average delay problem; and (iii) we provide a methodology for solving the delay minimization problem in settings consisting of tandem and parallel queues with heterogeneous servers.

We use the following notations throughout the paper. Non-bold lowercase and uppercase letters are used for scalers and sets, respectively. Bold lowercase letters are used for vectors. In addition, $\mathds{E}[.]$ denotes the expectation operator. The sub-6 GHz and mmWave variables are denoted by $\left(\cdot\right)_\text{sub-6}$ and $\left(\cdot\right)_\text{mm}$, respectively. 

\section{Related Work}
\label{sec:related-work}
In order to mitigate the effects of blockers and intermittent mmWave links, there have been several works to provide reliable communication over the mmWave band.
In \cite{genc2010robust}, the authors utilized reflections from walls or reflectors to assist directional paths. 
However, the behavior of the reflectors depends on the relative placement of transmitter/receiver to reflectors (e.g., incident angle). In addition, there is an additional reflection loss incurred due to reflectors. In \cite{fonseca2006collection}, a spatial diversity technique was utilized to combat blockage  caused by human movement. The technique delivers the same packets through several propagation paths simultaneously instead of the strongest path. Although this method can mitigate the effect of blockage, it increases energy consumption. The authors in \cite{singh2009blockage} devised a multi-hop directional MAC protocol for mmWave indoor wireless personal area networks. The key idea is to go around the obstacle through multi-hop paths. 
Directional MAC protocol is also studied in \cite{niu2015blockage} where the authors  investigate a joint optimization over relay selection and spatial reuse so that network performance could be improved. This method is limited to two-hop relaying. 

In another related line of research, the \emph{slow-server} problem, in which the goal is to obtain a delay optimal scheduling policy in a queuing system with heterogeneous (i.e., fast and slow) servers, has been studied. In that scenario, the main question that has been answered is: \emph{should the slow server be laid aside or utilized occasionally?} The focus of this problem is the trade-off between waiting in queue and entering slow servers when fast servers are busy or unavailable. In this context, the mmWave link acts as the fast server that becomes unavailable if blockage occurs. The authors in \cite{larsen1983control} presented a M/M/2 queuing system with two heterogeneous servers and conjectured that the optimal policy for minimizing the average delay and expected total discounted delay in system is of the threshold-type. The conjecture was later confirmed in \cite{lin1984optimal}. Following this work, \cite{rykov2001monotone} extended the result to the system with multi-servers (i.e., more than two), and \cite{ozkan2014optimal} studied the delay minimization problem with different arrival and service processes. 
{ Our delay minimization problem differs from the aforementioned works in two key aspects: 
\begin{enumerate}
\item In our system architecture (see Fig. \ref{origin_model} and Fig. \ref{model}), packets that are scheduled to the mmWave link have to first go through a processing server for essential data processing. This makes our system a mix of the tandem and parallel queues. Further, the processing server and the mmWave queue constitute a tandem queue which is part of parallel queues. In this case, to prove optimality of the proposed threshold-type policy, we need to show the relationship between the resulting delays starting at states with the packet in processing server and the packet moved to the mmWave queue, where the two states cannot be collapsed. This makes the problem more complex than the traditional slow server problem.

\item Our mmWave interface includes both a server and a buffer, which together constitute the mmWave queue. In addition, we require packets in the mmWave queue to be impatient, meaning that packets can be reneged to sub-6 GHz for service. Note that the packets in the sub-6 GHz server cannot be sent back to the mmWave queue or the head buffer. Adding the flexibility of reneging to the mmWave packets introduces several new challenges such as: should packets be routed from the mmWave queue or the head buffer to the sub-6 GHz queue? Therefore, in addition to the trade-off between waiting in the head queue (entrance to the system) and entering the slow server which is investigated in the slow-server problem, our problem also considers the trade-off between waiting in the mmWave queue and entering the slow server and trade-off between dispatching packets from the head buffer and the mmWave queue.
\end{enumerate}
}

\section{Problem Setup}
\label{sec:system-model}
In this section, we present the system model and formulate the delay minimization problem. 
\subsection{System Model}
We consider an integrated communication architecture with dual sub-6 GHz and mmWave interfaces as shown in Fig. \ref{origin_model}. 
The infinite \emph{head buffer} is utilized to store all packets waiting to be processed and served by either mmWave or sub-6 GHz. The \emph{processing server} is responsible for essential data processing before scheduling. 
Plus, the system includes two servers (mmWave and sub-6 GHz servers) with extremely different service rate, i.e. the service rate of mmWave can be 100 times larger than the service rate of sub-6 GHz. 

\textbf{(i) Queue Models:} In our system model, we add a buffer to the mmWave server, which stores packets routed from the head buffer. The rationality of our design (i.e., a separate queue for the mmWave interface) is described next. The service rate of the mmWave server is comparable to the processing server (i.e., processor speed). Moreover, mmWave is very sensitive to blockage, which is hard to quickly predict.  
If we assume that there is no buffer for the mmWave server, then every packet needs to wait in the head queue util the mmWave server is available. In the case, the packet will experience service time of both the processing and the mmWave servers (almost double the service time of the mmWave) except waiting time in the head buffer. Then, the performance of mmWave is degraded by approximately half. On the contrary, if the mmWave server has its own buffer for processed packets, part of waiting time in the head buffer can be utilized to process packets in advance, which reduces the experienced service time mentioned above. 
However, the sub-6 GHz link is much slower than the  processing server. Therefore, processing delay can be ignored compared to service time of the sub-6 GHz. In other words, it is not necessary for the sub-6 GHz server to have its own buffer considering the cost of buffer.
Thus, it is appropriate to assume that the \emph{sub-6 GHz interface} acts as a server with a buffer size of one, while the \emph{mmWave interface} consists of an infinite buffer and a server. 
\begin{figure}[t]
 \centering
 \includegraphics[scale=.35]{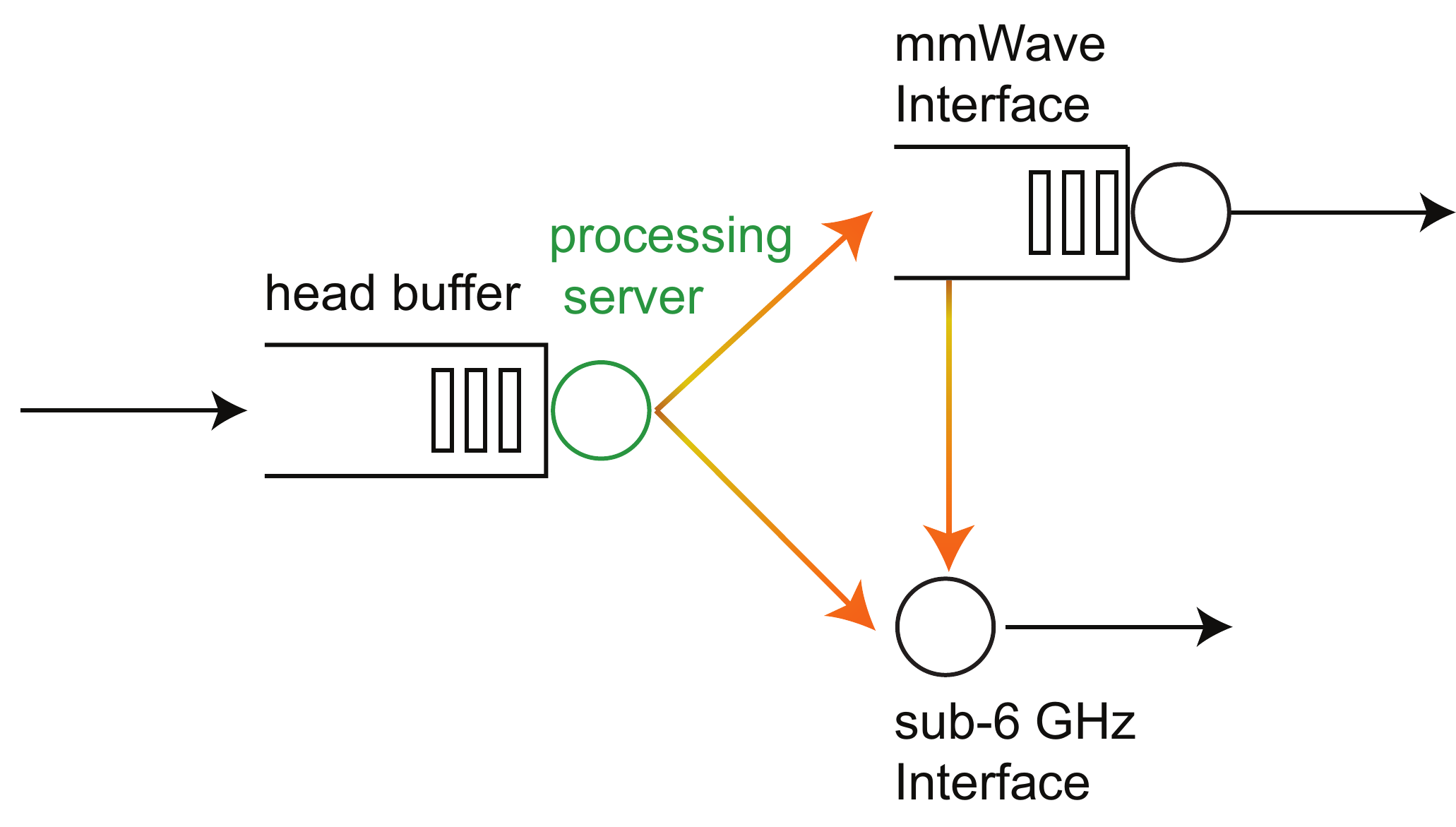}
 \caption{Integrated sub-6 GHz and mmWave architecture.}
 \label{origin_model}
  \vspace{-.3cm}
 \end{figure}

\textbf{(ii) Two-state mmWave link; Available or Unavailable:} As mentioned before, the mmWave link is highly variable with intermittent ON-OFF periods. 
It is reasonable to model the mmWave service rate with two states, say \emph{available} and \emph{unavailable}. For the unavailable state, the mmWave channel is almost disconnected and thus we assume that the service rate of the mmWave is 0. For the available state, we assume that the service time is exponentially distributed with parameter $\mu_\text{mm}$. Further, we denote the probability of available and unavailable states with $p_\text{a}$ and $p_\text{na}$, respectively. 

We further assume that arrivals to the system form a Poisson process with parameter $\lambda$, and that service times of the processing server and the sub-6 GHz interface are exponentially distributed with parameter {$\mu_\text{p}$} and $\mu_\text{sub-6}$, respectively. 
Given that the mmWave service rate is of the same order as the clock speed of the processor (i.e., several GHz), we assume that $\mu_\text{p}$ is much faster than $\mu_\text{sub-6}$ but in the same order as $\mu_\text{mm}$.  
Since delay of the processing server becomes negligible compared with the sub-6 GHz interface, we consider the equivalent model depicted in Fig. \ref{model} where we call the processing server and mmWave interface as \emph{mmWave line.} 

Within this content, we further clarify the difference of our problem from previous work, which has been briefly discussed in Section \ref{sec:related-work}. In Fig. 3, packets that are scheduled to the mmWave line have to go through a processing server first. These make our system a mix of the tandem and parallel queues. 
In the case, to finally obtain the optimality of the proposed threshold-type policy,
 we need to show the relationship between the resulting delays starting at states with the packet in processing server and the packet moved to the mmWave queue, where the two states cannot be collapsed at this step. This implies that our problem is more complex than the classic slow server problem.

\begin{figure}[t]
\centering
	\includegraphics[scale=.35]{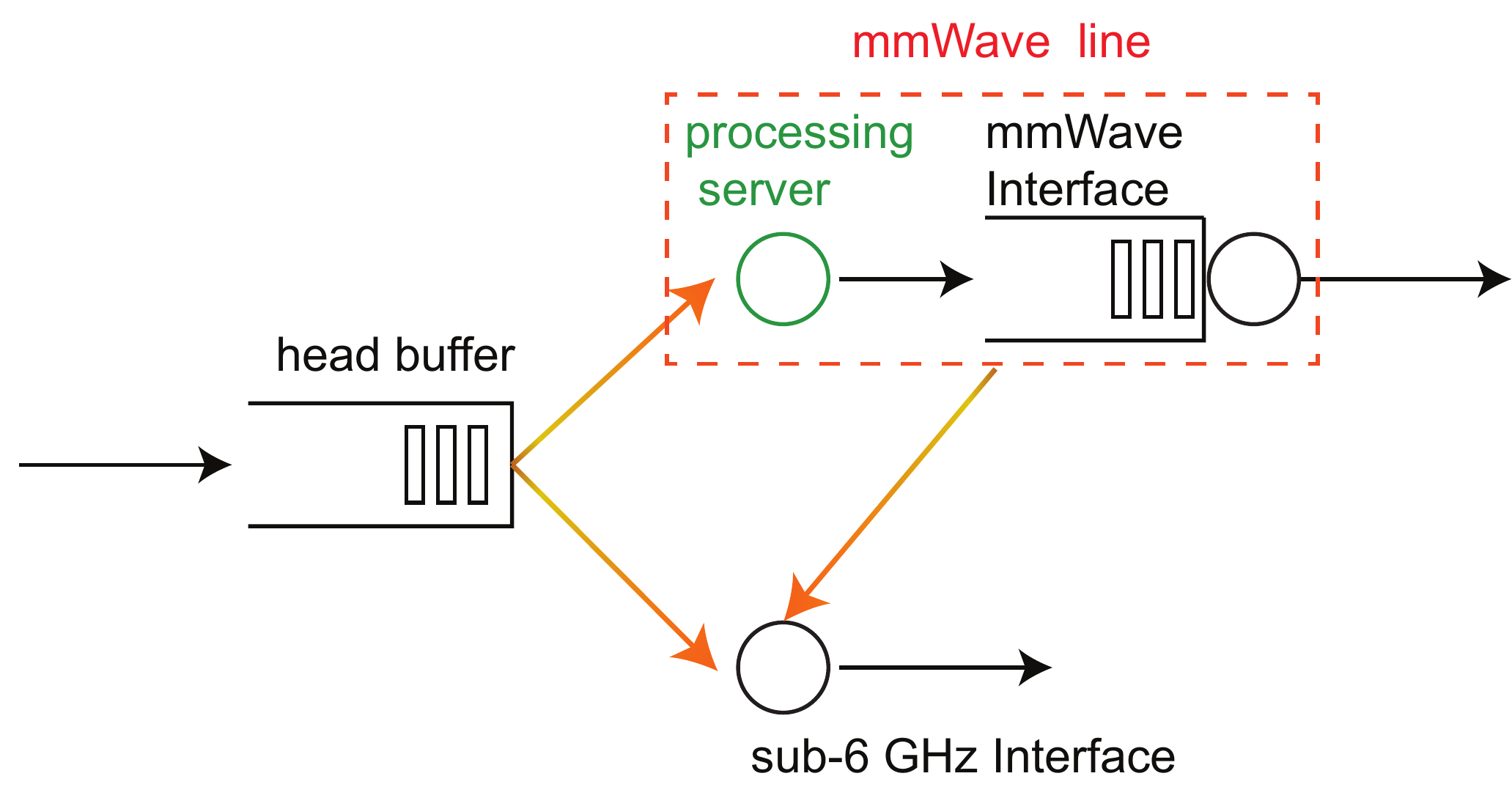}
	\caption{Equivalent system model.}
	\label{model}
	\vspace{-.4cm}
\end{figure}
As mentioned in Section \ref{sec:related-work}, to avoid a large waiting time in the mmWave queue due to intermittent channel (e.g., due to blockage), we require the packets to be \emph{impatient} in the sense that if the waiting time of the head-of-line packet in the mmWave queue becomes large, the packet ``reneges'' (is moved to) from the mmWave line or ``routes'' from the head buffer (is dispatched to) to the sub-6 GHz interface. 
Note that the packet in the sub-6 GHz server cannot be sent back to the mmWave line or the head buffer. Adding the reneging concept introduces new challenges such as : \emph{should the packets be moved from the head buffer or the mmWave queue to the sub-6 GHz server?} Therefore, in addition to the trade-off between waiting in the head queue and entering the slow server which is investigated in the slow-server problem, our problem investigates: (i) the trade-off between waiting in the mmWave line and entering the slow server, and (ii) trade-off between dispatching packets from the head buffer and the mmWave line.

\subsection{System Dynamics}
\label{notations}
\textbf{(i) System States:} 
Let $q_0$, $q_1$ $\in \mathbb{N}$ denote the queue length of the head buffer and mmWave interface, respectively. Moreover, $l_1$, $l_2$ $\in \{0,1\}$ denote the busy/idle condition of the processing server and sub-6 GHz interface, respectively. In this case, $l_1=1$ implies a busy server. Therefore, the system state can be expressed by a four-dimensional vector $\mathbf{q} \triangleq \left(q_0, l_1, q_1, l_2\right)$ with the state space of $Q\triangleq \mathbb{N}\times\{0,1\} \times\mathbb{N}\times\{0,1\}$.  

\textbf{(ii) Events:}
There are four different events that happen in the system, which are defined as follows:

\noindent \emph{(1) Arrival of a packet to the head buffer:} After arrival of one packet, state of the system is changed as follows:
\begin{align}
	\mathcal{A}_0\left(\mathbf{q}\right)
	\triangleq\left(q_0+1, l_1, q_1, l_2\right).\notag
\end{align}
 
\noindent \emph{(2) Departure of a packet from the mmWave interface:} The departure of a packet from the mmWave queue changes the system state as 
\begin{align}
	\mathcal{D}_1\left(\mathbf{q}\right)
	\triangleq\left(q_0, l_1, \left(q_1-1\right)^{+}, l_2\right),\notag
\end{align}
where $\left(\cdot\right)^+=\max\left(\cdot,0\right).$ 

\noindent \emph{(3) Departure of a packet from the sub-6 GHz interface:} If a packet departs from the sub-6 GHz queue, state of the system changes as
\begin{align}
	\mathcal{D}_2\left(\mathbf{q}\right)	\triangleq\left(q_0, l_1, q_1, \left(l_2-1\right)^{+}\right).\notag
\end{align} 

\noindent \emph{(4) Processing completion:} 
If the processing server delivers a packet to the mmWave queue, the system state changes as
\begin{align}
	\mathcal{T}\left(\mathbf{q}\right)
	\triangleq\left(q_0,\left(l_1-1\right)^{+}, l_1+q_1, l_2\right).\notag
\end{align} 
Note that we introduce ``dummy'' packets for the last three events when $q_1=0$, $l_2=0$ and $l_1=0$, respectively. This is further elaborated in Section \ref{formulation}.  

\textbf{(iii) Actions:} $K=\{A_h, A_1, A_2, A_b, A_r\}$ is an action set. $K_{\mathbf{q}} \subseteq K$ denotes the set of admissible actions in state $\mathbf{q}$. Each action in set $K$ is defined as follows: 

\noindent \emph{(1) Holding:} Action $A_h$ keeps the system state unchanged, and is defined on $Q$. Therefore, we have 
$$A_h\left(\mathbf{q}\right)\triangleq\left(q_0, l_1, q_1, l_2\right).$$

\noindent \emph{(2) Scheduling-on-mmWave:} 
A packet can be routed to the mmWave line if the processing server is idle, i.e., $$A_\text{1}\left(\mathbf{q}\right)\triangleq\left(q_0-1, 1, q_1, l_2\right),$$ 
which is defined on the set $\{\mathbf{q}\mid q_0\geq 1, l_1=0\}$. 

\noindent \emph{(3) Scheduling-on-sub-6:}  
A packet can be routed to the sub-6 GHz interface if the sub-6 GHz server is idle, i.e., 
$$A_2\left(\mathbf{q}\right)\triangleq\left(q_0-1, l_1, q_1, 1\right),$$ 
which is defined on the set $\{\mathbf{q}\mid q_0\geq 1, l_2=0\}$.  

\noindent \emph{(4) Scheduling-on-both:}
Action $A_b$ dispatches two packets to the sub-6 GHz and processing servers simultaneously, i.e., $$A_b\left(\mathbf{q}\right)\triangleq\left(q_0-2, 1, q_1, 1\right),$$ which is defined on the set $\{\mathbf{q} \mid q_0\geq 2, l_1=l_2=0\}$. 

\noindent \emph{(5) Reneging:} Action $A_r$ moves a packet from the mmWave line to the sub-6 GHz interface, and it is defined on the set $\left\{\mathbf{q}\mid q_1+l_1\geq 1, l_2=0 \right\}$. 
 Let $A_{r_\text{p}}$ and $A_{r_\text{mm}}$ denote the reneging actions from the processing server and mmWave interface, respectively. Therefore, we have
\begin{align}
\begin{array}{l l}
A_{r_\text{p}}\left(\mathbf{q}\right)\triangleq\left(q_0, 0, q_1, 1\right), &\mathbf{q}\in \{\mathbf{q}\mid l_1=1,\, l_2=0\}	;\\
A_{r_\text{mm}}\left(\mathbf{q}\right)\triangleq\left(q_0, l_1, q_1-1, 1\right),&\mathbf{q}\in \{\mathbf{q}\mid q_1\geq1,\, l_2=0\}.	
\end{array}\notag
\end{align} 
Then, the reneging action $A_r$ is expressed as 
\begin{align}
& A_r\left(\mathbf{q}\right)\triangleq 
 \left\{
\begin{array}{l l}
A_{r_\text{p}}\left(\mathbf{q}\right) &\text{if}\,\,l_1=1,\,\,q_1=0\\
A_{r_\text{mm}}\left(\mathbf{q}\right)&\text{if}\,\,l_1=0,\,\,q_1\geq1\\
\mathop {\argmin }\limits_{A_a\in \{A_{r_\text{p}}, A_{r_\text{mm}}\}} v\left(A_a\left(\mathbf{q}\right)\right)&\text{otherwise}
\end{array} \right. \notag	
\end{align}
where $v\left(\cdot\right)$ denotes the delay cost. Note that if $A_{r_\text{p}}$ and $A_{r_\text{mm}}$ are admissible, we select an action that results in a smaller cost. In Section \ref{sec:optimal-discounted-delay}, we show that $A_r=A_{r_\text{p}}$ for the discounted delay problem when both $A_{r_\text{p}}$ and $A_{r_\text{mm}}$ are admissible. 

\subsection{Problem Formulation}
\label{formulation}
\textbf{Average Delay Problem:}  Our objective is to schedule packets across the mmWave and sub-6 GHz interfaces such that the average delay of system is minimized. To this end, we know that, by Little's Law, the average delay minimization problem is equivalent to minimizing the average total number of packets in the system, which is expressed as follows:
\begin{equation}
\min_{\pi \in \Pi}\,\,\limsup_{T \to \infty}\frac{1}{T}\mathbb{E}^\pi\left[\int_{t=1}^{T}\left(\mathbf{q}[t]\cdot\mathbf{e}\right) dt\right],
\end{equation}
where $\mathbb{E}^\pi$ denotes the conditional expectation given policy $\pi$, $\mathbf{q}[t]\in Q$ is the system state at time $t$, $\Pi$ denotes the set of all admissible policies, and $\mathbf{e} = \left(1, 1, 1, 1\right)^{\text{T}}$. We model the system evolution as an MDP, and for simplicity, we convert the continuous-time MDP problem into an equivalent discrete-time MDP problem with the method of uniformization \cite{puterman2014markov}. In particular, we assume that all servers will serve ``dummy'' packets whenever they are idle. Then, we separate continuous time into time slots with sequences when either a packet arrival or a packet (real or dummy) departure from the processing server or interfaces happens. 
Let $N=\{1,2,3,\cdots\}$ denote the set of time slots such that the channel state does not change during each time slot. Then, the system state at the $n$-th time slot is expressed as $\mathbf{q}[n]$.
Furthermore, without loss of generality, we scale time and assume that $\lambda+\mu_\text{p}+p_\text{a}\mu_\text{mm}+\mu_\text{sub-6}=1$.

We consider the set of control variables $U\triangleq \{(u_0, u_1, u_2,\\ u_3)\mid u_0, u_1, u_2, u_3\in K\}$. 
Then, the decision rule at the $n$-th decision epoch (the beginning of the $n$-th time slot) is a mapping from the system states to the control variables, i.e., $d_n: Q \rightarrow U$, for all $ n\in N$ and the policy $\pi$ is a sequence of the decision rules, i.e., $\pi = \left(d_1, d_2, \cdots\right)$.
Further, if $\mathbf{q}\left[n\right]=\mathbf{q'}$ and $d_n\left(\mathbf{q'}\right)=\left(u_0, u_1, u_2, u_3\right)$ for certain $n \in N$, 
then if an arrival occurs at the $(n+1)$-th epoch, we would take actions according to $u_0$. Similar explanation applies to $u_1$, $u_2$, $u_3$. Thus, the transition probabilities in the discrete-time MDP are expressed as
\begin{align}
	\mathbb{P}\left(\mathbf{q}^{\prime}\mid \mathbf{q},\mathbf{u}\right)=
	\begin{cases}
		\lambda & \text{if}\,\,\, \mathbf{q}^{\prime}=u_0\left(\mathcal{A}_0\left(\mathbf{q}\right)\right)\\
		\mu_\text{p}& \text{if}\,\,\, \mathbf{q}^{\prime}=u_1\left(\mathcal{T}\left(\mathbf{q}\right)\right)\\
		p_\text{a}\mu_\text{mm} & \text{if}\,\,\, \mathbf{q}^{\prime}=u_2\left(\mathcal{D}_1\left(\mathbf{q}\right)\right)\\
		\mu_\text{sub-6} & \text{if}\,\,\, \mathbf{q}^{\prime}=u_3\left(\mathcal{D}_2\left(\mathbf{q}\right)\right)\\
	\end{cases}\notag
	\vspace{-0.2cm}
\end{align}
Then, with the discrete-time MDP, the uniformized problem is formulated as follows:
\begin{equation}
\vspace{-0.05cm}
\min_{\pi \in \Pi}\,\,\limsup_{N \to \infty}\frac{1}{N}\mathbb{E}^\pi\left[\sum_{n=1}^{N}\mathbf{q}[n]\cdot\mathbf{e}\right].
\label{eq:avg-delay}
\vspace{-0.05cm}
\end{equation}

\textbf{Discounted Delay Problem:}
To solve the average delay problem, we first consider the problem of minimizing the expected total discounted delay of the system (discounted delay problem) to avoid convergence issues in the presence of bounded value function \cite{puterman2014markov}. Next, we extend our results to the average delay problem. The discounted delay problem in the equivalent discrete-time MDP is expressed as 
\begin{equation}
\min_{\pi \in \Pi}\,\,\,\,\, \mathbb{E}^\pi \left[ \sum_{n=1}^{\infty}\beta^{n-1}\mathbf{q}[n]\cdot\mathbf{e}\right],
\end{equation}
where $\beta$ is a discount factor such that $0\leq\beta<1$. To solve the discounted delay problem, it is known that there exists an optimal deterministic stationary policy \cite{puterman2014markov}. Thus, we only need to consider the class of deterministic stationary policies.  
We apply the value iteration method to find the optimal policy.

Under the assumption that the system is stable, value (delay) functions of the initial state $\mathbf{q}\in Q$ are bounded real-valued functions. Let $V$ denote the Banach space of bounded real-valued functions on $Q$ with supremum norm. Define operator $\mathcal{L}: V \rightarrow V$ as
\begin{IEEEeqnarray}{rCl}
  \IEEEeqnarraymulticol{3}{l}{
   \left(\mathcal{L}v\right)\left(\mathbf{q}\right)
  }\nonumber\\\quad
  & \triangleq & \mathbf{q}\cdot\mathbf{e}+\beta\min_{\mathbf{u}\in U_{\mathbf{q}}} \bigg\{\lambda  v\big(u_0\left(\mathcal{A}_0\left(\mathbf{q}\right)\right)\big)
	 +\mu_\text{p}v\big(u_1\left(\mathcal{T}\left(\mathbf{q}\right)\right)\big)\notag\\
  && +\mu_\text{mm}p_\text{a} v\big(u_2\left(\mathcal{D}_1\left(\mathbf{q}\right)\right)\big)+\mu_\text{sub-6} v\big(u_3\left(\mathcal{D}_2\left(\mathbf{q}\right)\right)\big)\bigg\},
	\label{eq:optimality-condition}
\end{IEEEeqnarray}
where $v\left(\cdot\right)\in V$ and $U_\mathbf{q}$ denotes the set of admissible control variables in state $\mathbf{q}$ such that $U_\mathbf{q}\subseteq U$. 
Let $J_\beta\left(\mathbf{q}\right)$ denote optimal expected total discounted delay function of initial state $\mathbf{q}$. Then, $J_\beta\left(\mathbf{q}\right)$ is a solution of Bellman function, i.e., $J_\beta\left(\mathbf{q}\right)=\mathcal{L}J_\beta\left(\mathbf{q}\right)$. 
\section{Delay Optimal Policy}
\label{sec:optimal-discounted-delay}
\subsection{Discounted Delay Problem}
 Except that the mmWave channel is extremely intermittent, the average service rate of the mmWave is much higher than the sub-6 GHz (e.g., two orders of magnitude). Besides, the service rate of the mmWave and processing server are in the same order. Hence, it is reasonable to assume that the expected time for a packet to go through empty mmWave line is less than empty sub-6 GHz interface, i.e., $\frac{1}{p_\text{a} \mu_\text{mm}}+\frac{1}{\mu_\text{p}}<\frac{1}{\mu_\text{sub-6}}$. 
With this assumption, we have the following theorem: 
\begin{theorem}
\label{thm:optimality}
Assuming that \emph{$\frac{1}{p_\text{a} \mu_\text{mm}}+\frac{1}{\mu_\text{p}}<\frac{1}{\mu_\text{sub-6}}$}, then we have
\emph{ 
\begin{equation}
\begin{array}{lcl}
\text{(a)} \  J_\beta(A_1(\mathbf{q}))\leq J_\beta(A_{h}(\mathbf{q})) & & \text{if}\,\,q_0\geq 1, \,\,l_1=0;\notag\\
\text{(b)}  \ J_\beta(A_{2}(\mathbf{q}))\leq J_\beta(A_r(\mathbf{q})) & &\text{if}\,\,q_0\geq 1, \,\,l_1+q_1\geq1,	\notag\\
& & \,\,\,\,\,\,\text{and} \,\,l_2=0;\notag\\
\text{(c)}\ J_\beta(\mathcal{T}(\mathbf{q}))\leq J_\beta(\mathbf{q}) & &\text{if} \,\,l_1=1;\notag\\ 
\text{(d)}  \ J_\beta(A_1(\mathbf{q}))\leq J_\beta(A_{2}(\mathbf{q})) & & \text{if}\,\,\mathbf{q}=(q_0,0,0,0)\notag\\
& &  \,\,\,\,\,\,\text{and}\,\, q_0\geq 1 ;\notag\\ 
\text{(e)} \ J_\beta(\mathbf{x})\leq J_\beta(\mathbf{y}) & & \text{if}\,\,\Vert \mathbf{x} \Vert_1 \leq \Vert \mathbf{y} \Vert_1, \,\,\mathbf{x}, \mathbf{y} \in Q.  
\end{array}  	
\end{equation}}
\end{theorem}
\begin{proof}
Proof is provided in Appendix A. 
\end{proof}

\noindent\emph{Remark:} Note that in the following, if action $A_{x}\in K$ has a higher priority than action $A_{y}\in K$, it means that action $A_{x}$ incurs no more costs than action $A_{y}$, where $x,y\in \{1,2,r,b,h\}$. 

From Theorem \ref{thm:optimality}, we obtain three rules that provides partial characteristics of the optimal policy:
\begin{basedescript}{\desclabelstyle{\pushlabel}\desclabelwidth{3.1em}}
\item[\emph{Rule 1.}] \emph{Holding is not preferable as long as the processing server is idle:} Property (a) implies that action $A_1$ has priority over action $A_h$. 
\item[\emph{Rule 2.}] \emph{Keeping the mmWave line busy:} Properties (a) and (d) imply that a packet should be scheduled on the mmWave line whenever the mmWave line is empty and the head buffer  (see Fig. \ref{model}) is not empty.
\item[\emph{Rule 3.}] \emph{Head buffer is the first choice for the sub-6 GHz interface:} 
By property (b), action $A_2$ has priority over action $A_r$. 
In addition, $J_\beta(\mathcal{T}(\mathbf{q}))=J_\beta(A_{r_p}(\mathbf{q}'))$ and $J_\beta(\mathbf{q})=J_\beta(A_{r_{mm}}(\mathbf{q}'))$, where $\mathbf{q}=(q_0,1,q_1,1)$ and $\mathbf{q}'=(q_0,1,q_1+1,0)$. Then, property (c) implies that $A_r(\mathbf{q}')=A_{r_\text{p}}(\mathbf{q}')$ for $A_{r_\text{p}}$, $A_{r_\text{mm}} \in K_{\mathbf{q}'}$.\end{basedescript}
 

\textbf{Optimal Policy:} Based on these rules, we show that optimal policy for the discounted delay problem is of the threshold-type, and is defined as follows:
\begin{align}
& D_{m}\left(\mathbf{q}\right)=\notag\\
& 
  \begin{cases}
  A_1\left(\mathbf{q}\right) & \text{if} \,\, \mathbf{q}=\left(q_0,0,q_1,1\right),\,q_0\geq 1,\notag\\
  &\ \ \, \text{or}\, \, \mathbf{q}=\left(q_0,0,q_1,0\right),\,q_0\geq 1,\,\, q_0+q_1\leq m, \\     
  A_2\left(\mathbf{q}\right) & \text{if} \,\, \mathbf{q}=\left(q_0,1,q_1,0\right),\, q_0\geq 1,\,\, q_0+q_1+1>m,\notag\\
  & \ \ \, \text{or}\,\, \mathbf{q}=
  \left(1,0,q_1,0\right),\,q_1\geq m, \\  
  A_{r}\left(\mathbf{q}\right) & \text{if}\, \, \mathbf{q}=\left(0,l_1,q_1,0\right),\, l_1+q_1>m,  \\  
  A_{b}\left(\mathbf{q}\right) & \text{if}\, \, \mathbf{q}=\left(q_0,0,q_1,0\right),\, q_0+q_1>m, q_0\geq 2, \\  
  A_{h}\left(\mathbf{q}\right) & \text{otherwise}, 
  \end{cases} 
\end{align}
where $D_{m}$ is a threshold policy with threshold $m$ such that $D_m$ follows all above rules. Then, for each $n \in N$, the decision rule at time slot $n$ is given by
$d_n\left(\mathbf{q}\right)=\left(D_m\left(\mathcal{A}_0\left(\mathbf{q}\right)\right),D_m\left(\mathcal{T}\left(\mathbf{q}\right)\right),D_m\left(\mathcal{D}_1\left(\mathbf{q}\right)\right),D_m\left(\mathcal{D}_2\left(\mathbf{q}\right)\right)\right)$. 

 To prove the optimality of $D_{m}$ for the discounted delay problem, we name the action sets $\{A_1,A_h\}$ and $\{A_2,A_r\}$ as \emph{``not-adding-to-sub-6''} and exclusively \emph{``adding-to-sub-6''}, respectively. We already know the priority between $A_1$ and $A_h$ and the priority between $A_2$ and $A_r$. Thus, it only remains to determine the priority between the sets not-adding-to-sub-6 and adding-to-sub-6. To show this, we dub the path consisting of the head buffer, the processing server, and the mmWave queue as ``\emph{FastLane}''. We claim that in the discounted delay optimal policy, adding-to-sub-6 obtains priority over not-adding-to-sub-6 when the queue length of FastLane exceeds certain threshold $m$, i.e., a threshold-type policy as expressed by $D_m$. Next, we show this via value iteration. For simplicity, we re-express the system state $\mathbf{q}$ in the form of $\left(x,q_1,l_2\right)$ where $x$ denotes the number of packets in the head buffer and processing server. Note that if $x>0$, then the processing server should be busy by Rule 1.
For the sake of exposition in the following proof, we define two terms in Definition \ref{intermediate value}.
\begin{definition}
\label{intermediate value}
Let $J_\beta^n\left(x,q_1,l_2\right)$ denote the optimal expected total discounted delay over the next $n$ time slots with initial state $\left(x,q_1,l_2\right)$. Then, we define an intermediate value $T_\beta^n\left(x,q_1,l_2\right)$ as:\emph{
\begin{align}
& T_\beta^n\left(x,q_1,l_2\right)=\notag\\
&\left\{
\begin{array}{l l}
	J_\beta^n\left(x,q_1,l_2\right) &\text{\emph{if}}\,\,\mathbf{q}=\mathbf{0}\,\,\text{\emph{or}}\,\, l_2=1\\ 
	\min\{J_\beta^n\left(x,q_1,0\right),J_\beta^n\left(x-1,q_1,1\right)\}&\text{\emph{if}}\,\,x\geq 1, l_2=0\\
	\min\{J_\beta^n\left(0,q_1,0\right),J_\beta^n\left(0,q_1-1,1\right)\}&\text{\emph{otherwise}}
 \end{array}\right.\notag
\end{align}	}
\noindent As a result, $J_\beta^{n+1}\left(x,q_1,l_2\right)$ is written as:
\emph{\begin{align}
	 &J_\beta^{n+1}\left(x,q_1,l_2\right)\notag
	 =  \left(x+q_1+l_2\right)+
	 \beta\lambda T^n_\beta\left(x+1,q_1,l_2\right)\notag\\
	 &+\beta\mu_\text{mm} p_\text{a} T^n_\beta\left(x,\left(q_1-1\right)^+,l_2\right)+\beta\mu_\text{sub-6} T^n_\beta\left(x,q_1,0\right)\notag\\
	 	 &+\beta\mu_\text{p} T^n_\beta\left(\left(x-1\right)^+,x+q_1-\left(x-1\right)^+,l_2\right).
	 	 	\label{valueiteration}
\end{align}}
Moreover, $J_\beta^{0}\left(x,q_1,l_2\right)=x+q_1+l_2$.
\end{definition}
Next, we define a class of functions with threshold, supermodular and monotonicity properties in Definition \ref{F_class} and Lemma \ref{thm:F_class} proves that $J_\beta^n$ has these properties.

\begin{definition}
\label{F_class}
Let $\mathscr{F}$ be a class of functions such that for each function \emph{$f : \mathbb{N}\times \mathbb{N}\times \{0,1\} \to \mathbb{R}_{\geq 0}$} in \emph{$\mathscr{F}$}, we have 
\emph{\begin{align}
	& f\left(x+1,q_1,0\right)+f\left(x+1,q_1,1\right)\notag\\
	 &\ \ \ \ \ \ \ \,  \ \ \ \ \ \  \ \ \ \ \  \leq f\left(x,q_1,1\right)+f\left(x+2,q_1,0\right)\label{A2_x}\\
	& f\left(x+1,q_1,0\right)+f\left(x,q_1+1,1\right)\notag\\
	& \ \ \ \ \ \ \ \,  \ \ \ \ \ \  \ \ \ \ \  \leq f\left(x,q_1,1\right)+f\left(x+1,q_1+1,0\right)\label{A2_q1}\\
	& f\left(0,q_1+1,0\right)+f\left(0,q_1+1,1\right)\notag\\
	&\ \ \ \ \ \ \ \,  \ \ \ \ \ \  \ \ \ \ \ \leq f\left(0,q_1,1\right)+f\left(0,q_1+2,0\right)\label{Ar}\\
	& f\left(x,q_1+1,l_2\right)\leq f\left(x+1,q_1,l_2\right)\label{switch}	\end{align}}
together with supermodularity:
	\emph{\begin{align}
	& f\left(x,q_1,1\right)+f\left(x+1,q_1,0\right)\notag\\
			& \ \ \ \ \ \ \ \ \ \ \ \ \ \ \ \ \  \ \ \ \ \ \leq f\left(x,q_1,0\right)+f\left(x+1,q_1,1\right)\label{supermodular1}\\
	& f\left(x,q_1,1\right)+f\left(x,q_1+1,0\right)\notag\\
	& \ \ \ \ \ \ \ \ \ \ \ \ \ \ \ \ \  \ \ \ \ \ \leq f\left(x,q_1,0\right)+f\left(x,q_1+1,1\right)\label{supermodular2}
	\end{align}}
and monotonicity:
	\emph{\begin{align}
	& f\left(x,q_1,l_2\right)\leq f\left(x+1,q_1,l_2\right)\label{mono1}\\
	& f\left(x,q_1,l_2\right)\leq f\left(x,q_1+1,l_2\right)\label{mono2}\\
	& f\left(x,q_1,0\right)\leq f\left(x,q_1,1\right)\label{mono3}
\end{align} }
\end{definition}
\noindent Eq. \eqref{A2_x} to \eqref{Ar} describe the threshold property that is clarified in the proof of Lemma \ref{lemma1}.
\begin{lemma}
\label{thm:F_class}
The optimal expected total discounted delay over the next $n$ time slots $J^n_\beta$ satisfies all properties in Definition \ref{F_class}, i.e., $J^n_\beta\in \mathscr{F}$ for each $n \in \mathbb{N}$.  
\end{lemma}
\begin{proof}
Proof is provided in Appendix B. 
\end{proof}
\noindent Next, we use Lemma \ref{thm:F_class} to prove that each round of value iteration corresponds to a threshold-type policy as expressed by Lemma \ref{lemma1}. 

%
\begin{lemma}
\label{lemma1}
For each round of value iteration, the corresponding policy is of the threshold-type.
\end{lemma}
\begin{proof}
Proof is provided in Appendix C. 	
\end{proof}

\noindent Finally, we use Lemma \ref{lemma1} to provide our main result that the optimal policy is of the threshold-type. 
\begin{theorem}
\label{final_result}
For the discounted delay optimality problem, there exists an optimal stationary policy that is of the threshold-type with threshold $m\leq \infty$.
\end{theorem}
\begin{proof}
By Lemma 2, for each round of value iteration, corresponding policy is of threshold-type. Thus, as $n \to \infty$, the corresponding policy is also of the threshold-type, and the policy is expected total discounted delay optimal policy.
\end{proof}

\textbf{Optimal Threshold:} Theorem \ref{lemma2} proves that the value of the threshold in the optimal policy of each iteration, increases by at most one unit at the next iteration. 

\begin{theorem}
\label{lemma2}
If threshold value of the policy corresponding to $n$-th value iteration is $i_n$, then the policy corresponding to $n+1$-th value iteration has threshold value $i_{n+1}\in [0, i_n+1]$.	
\end{theorem}
\begin{proof}
We re-express the system state as $\left(y, l_2\right)$, where $y\in \mathbb{N}$ denotes the queue length of FastLane. 
Then, Lemma \ref{lemma1} is expressed as follows:\begin{align}
	&J^n_\beta\left(y+1,0\right)<J^n_\beta\left(y,1\right),\,\,\text{if}\,\, y\leq i_n-1;\\
	&J^n_\beta\left(y+1,0\right)\geq J^n_\beta\left(y,1\right),\,\,\text{if}\,\, y \geq i_n.\label{thr1}
\end{align}
Since $J^{n+1}_\beta\left(y+1,0\right)-J^{n+1}_\beta\left(y,1\right)$ increases with $y$, it remains to show that
$
J^{n+1}_\beta\left(y+1,0\right)-J^{n+1}_\beta\left(y,1\right)\geq 0
$, 
when $y \geq i_n+1$. In fact, 
\begin{align}
J^{n+1}_\beta \left(y+1,0\right)&-J^{n+1}_\beta\left(y,1\right)\notag\\
=&\beta\lambda \left(T^{n}_\beta\left(y+2,0\right)-T^{n}_\beta\left(y+1,1\right) \right)\notag\\
+& \beta\mu_\text{p} \left(T^{n}_\beta\left(y+1,0\right)-T^{n}_\beta\left(y,1\right)\right )+\beta\mu_\text{mm}p_\text{a}Z\notag\\
+&\beta\mu_\text{sub-6} \left(T^{n}_\beta\left(y+1,0\right)-T^{n}_\beta\left(y,0\right)\right),\notag 
\end{align}
where $Z=T^{n}_\beta\left(y,0\right)-T^{n}_\beta\left(y-1,1\right)$ or $Z=T^{n}_\beta\left(y+1,0\right)-T^{n}_\beta\left(y,1\right)$. Note that if $\mathcal{D}_1\left(y+1,0\right)=\left(y,0\right)$, then $\mathcal{D}_1\left(y,1\right)=\left(y-1,1\right)$. On the contrary, if we assume that $\mathcal{D}_1\left(y,1\right)=\left(y,1\right)$, then the only packet in the mmWave queue is reneged to the sub-6 GHz interface. This only happens when $y=0$ by the optimal policy, which contradicts with that $y\geq i_n+1$.
\noindent Since $y\geq i_n+1>i_n$, we have 
$$T^{n}_\beta\left(y+1,0\right)-T^{n}_\beta\left(y,1\right)\stackrel{\text{\eqref{thr1} }}{=}J^{n}_\beta\left(y,1\right)-T^{n}_\beta\left(y,1\right)=0.$$ 
Similarly, we obtain that  $T^{n}_\beta\left(y+2,0\right)-T^{n}_\beta\left(y+1,1\right)=0$ and $T^{n}_\beta\left(y,0\right)-T^{n}_\beta\left(y-1,1\right) =0$. 
As for $\mu_\text{sub-6}$ term, by monotonicity, we have $T^{n}_\beta\left(y+1,0\right)- T^{n}_\beta\left(y,0\right)\geq 0$.  
\end{proof}
\noindent\emph{Remark:} If we start with policy $D_{0}$ and the optimal threshold is $m^*$, then we can obtain the optimal threshold value in $m^*$ steps via policy iteration.
\subsection{Average Delay Problem}
\label{sec:optimal-average-delay}
The following theorem extends our results to the average delay problem. 

\begin{theorem}
There exists an optimal stationary policy of the threshold-type for the average delay problem.
\end{theorem}

\begin{proof}
According to \cite{lippman1973semi}, $\lim_{\beta_n\to 1}\left(1-\beta_n\right)J_{\beta_n}^{\pi_{\beta_n}^*}\left(\mathbf{q}\right)=J^{\pi^*}\left(\mathbf{q}\right)$, $\forall \mathbf{q}\in Q$, where $J_{\beta_n}^{\pi_{\beta_n}^*}\left(\mathbf{q}\right)$ denotes optimal expected total discounted delay under optimal policy $\pi_{\beta_n}^*$ associated with discount factor $\beta_n$ and $J^{\pi^*}\left(\mathbf{q}\right)$ denotes optimal average delay under optimal policy $\pi^*$. Since our action set is finite, by \cite{lippman1973semi}, there exists an optimal stationary policy for the average delay problem such that $\pi_{\beta_n}^*\to \pi^*$, which implies the optimal policy is of the threshold-type.
\end{proof}

In order to obtain the optimal threshold for the average delay minimization problem, we note that Theorem \ref{lemma2} also applies to this case as well, and the proof follows the same logic by removing the  discount factor $\beta$ in the proof of Theorem \ref{lemma2}.

\section{Simulation Results}
\label{sec:simulation}
In this section, we numerically investigate the performance of our proposed policy. To this end, we first investigate the relationship between the arrival rate and the optimal threshold. Next, we compare the performance of our policy against the MaxWeight policy. 

\subsection{Relationship between Arrival Rate and Optimal Threshold}
We investigate how the arrival rate $\lambda$ affects the optimal threshold of our policy. In simulations, we set $\mu_\text{mm}=\mu_\text{p}=100$, $\mu_\text{sub-6}=1$ and $p_\text{a}=0.6$. Then, we investigate how average delay changes as threshold varies given a value of $\lambda\in\{30,35,40,45,50,55\}$.
Our simulation results show that for $\lambda=30$, $35$, $40$, curves of average delay vs different threshold are similar. For lack of space, we only provide results for $\lambda=30$ here.

\begin{figure}[H]
\centering
 \subfloat[Arrival rate $\lambda=30$]{
 \includegraphics[scale=.21]{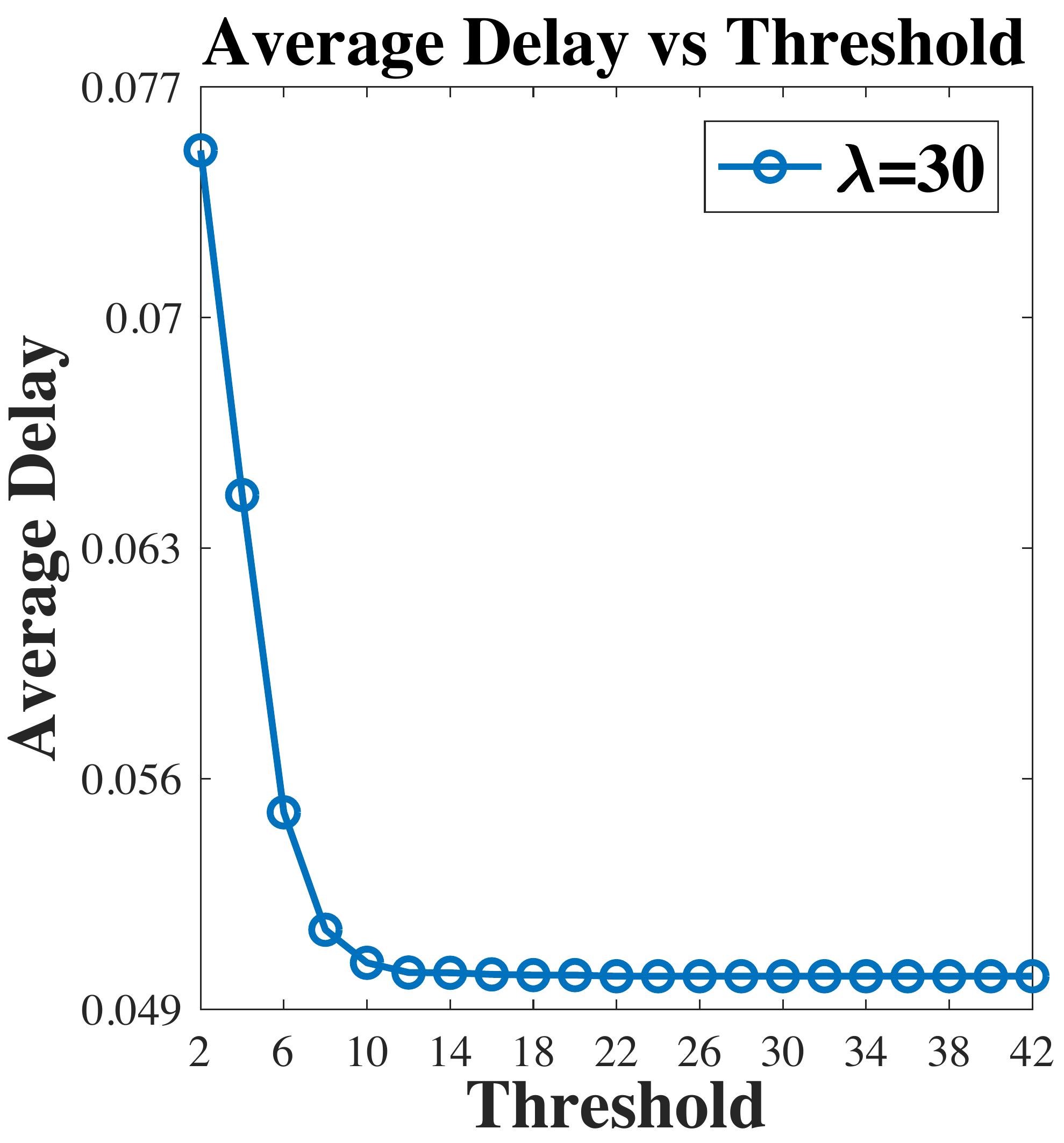}
 \label{MArrival Rate 30}
 }
 \subfloat[Arrival rate $\lambda=45$]{
 \includegraphics[scale=.21]{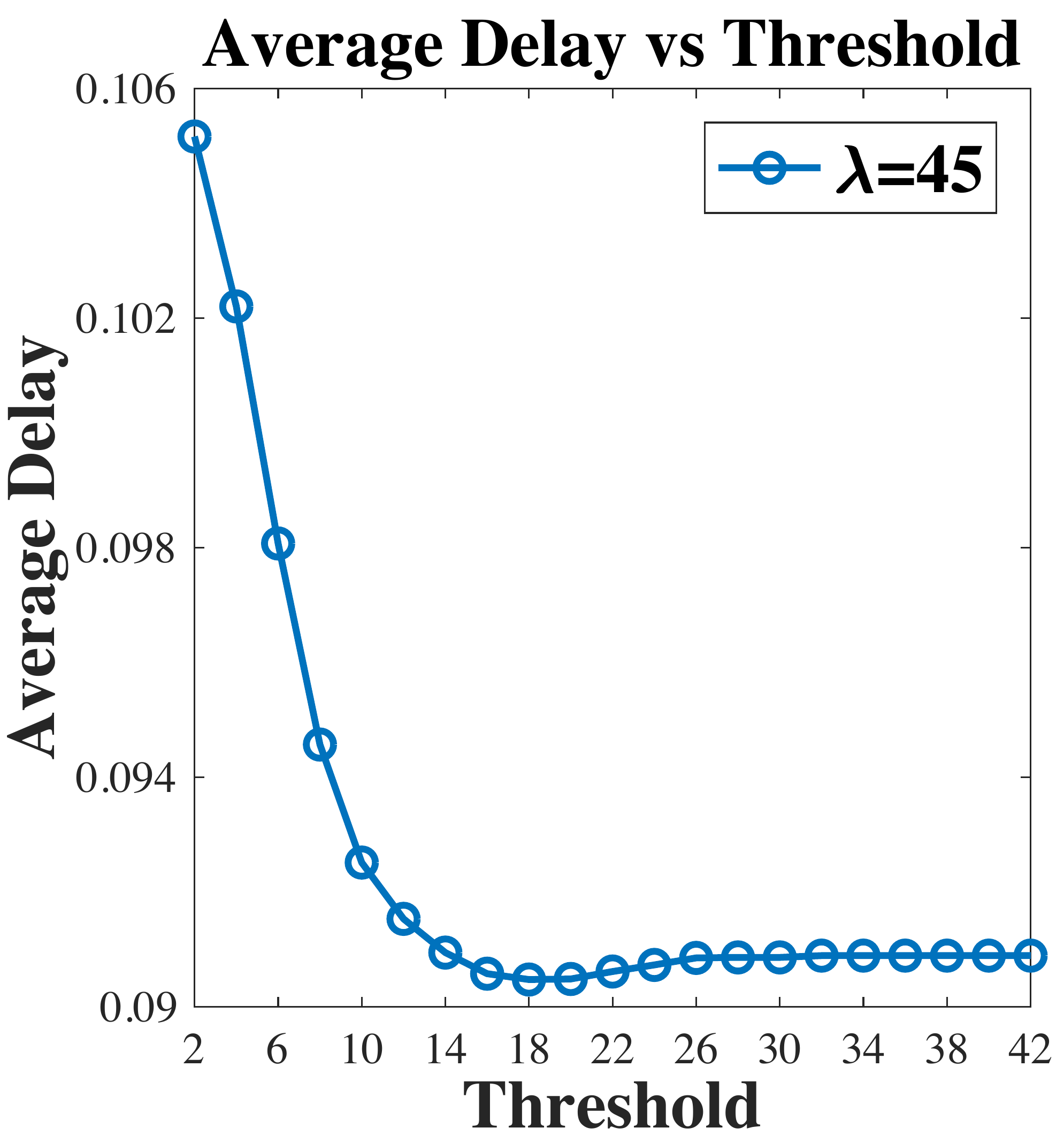}
 \label{MArrival Rate 45}
 }
 \\
 \centering
 \hspace{-.1cm}
 \subfloat[Arrival rate $\lambda=50$]{
 \includegraphics[scale=.21]{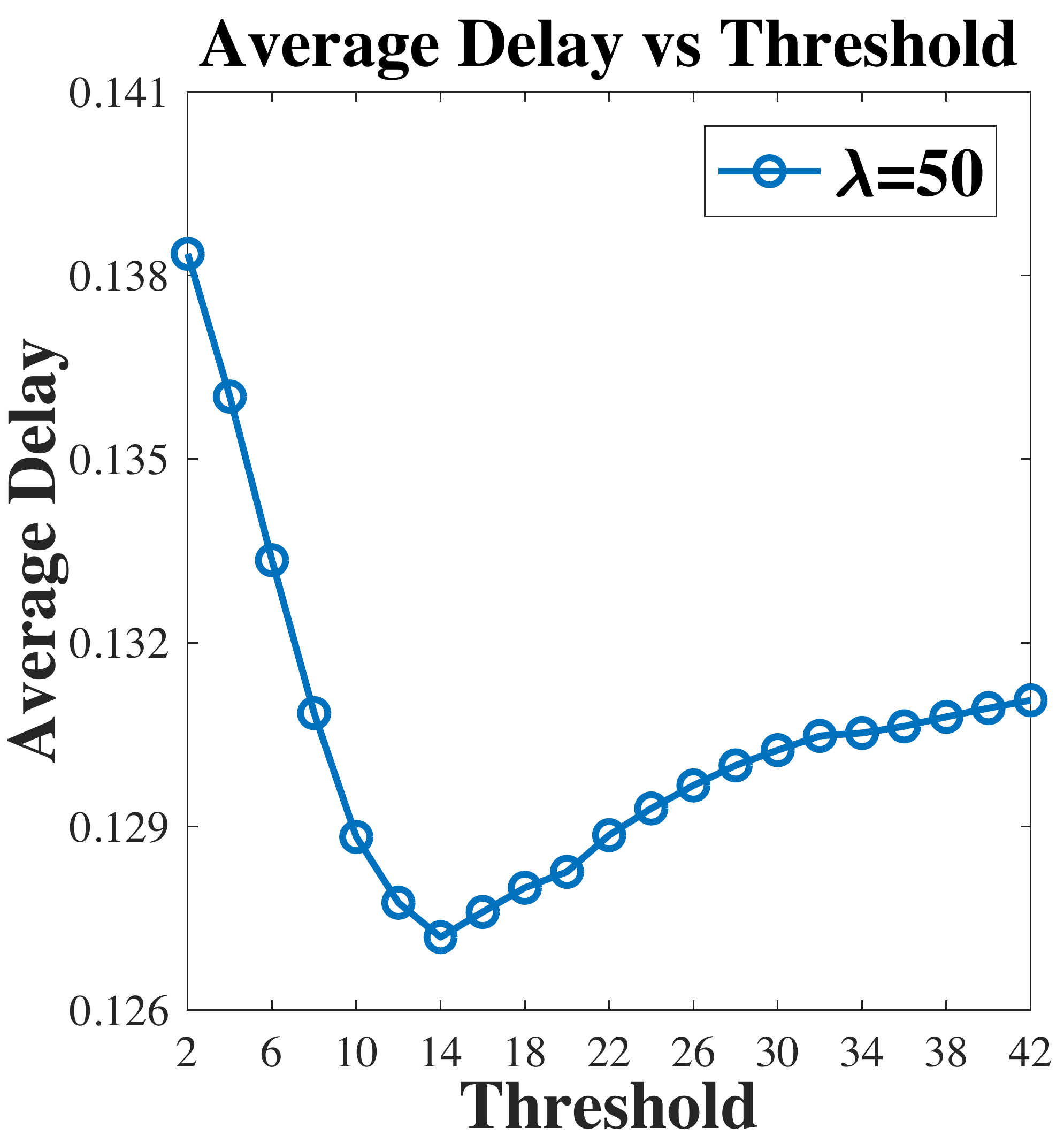}
 \label{MArrival Rate 50}
 }
 \subfloat[Arrival rate $\lambda=55$]{
 \includegraphics[scale=.21]{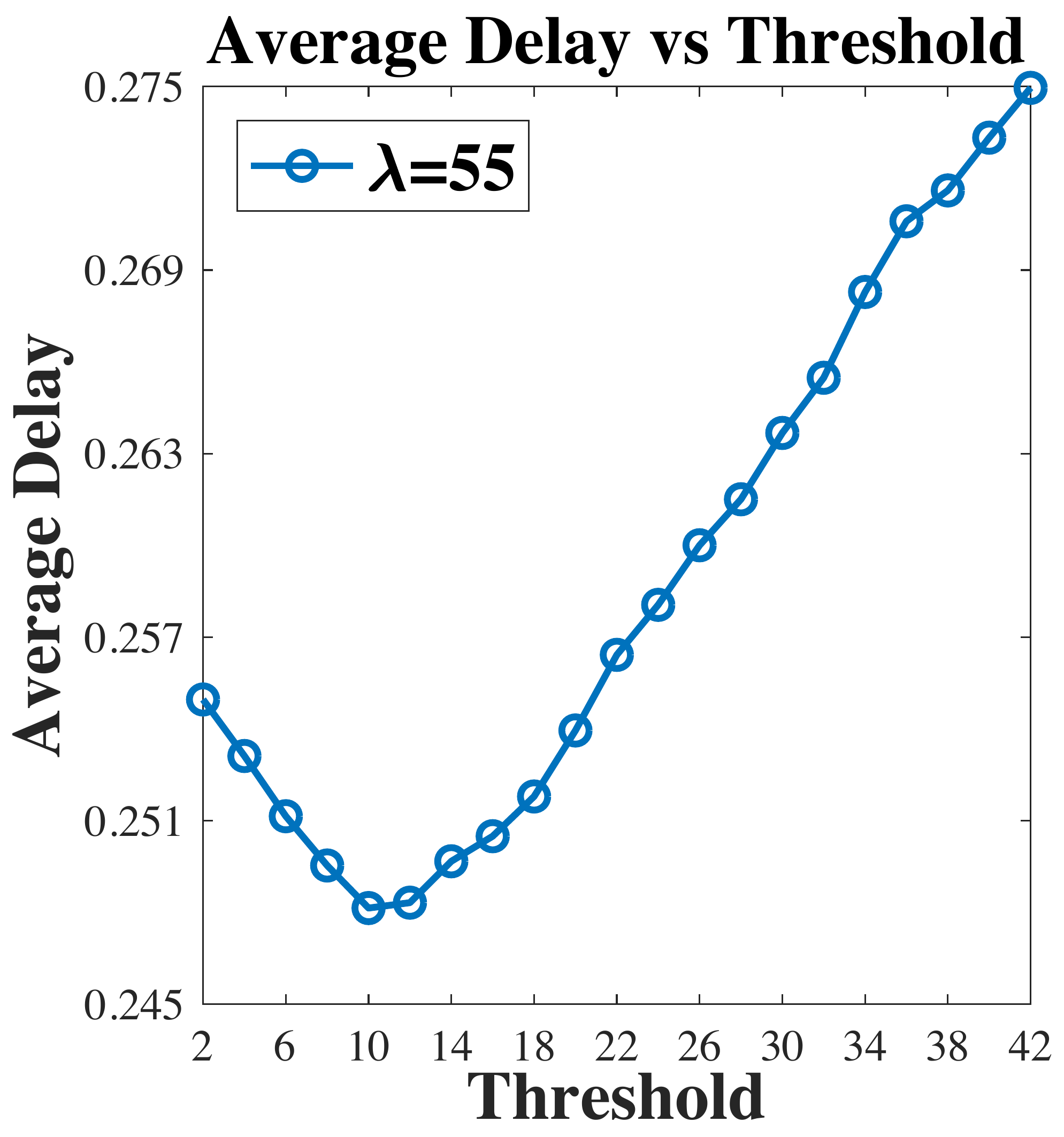}
 \label{MArrival Rate 55}
 }
 \caption{Average Delay vs Threshold for various arrival rate.}
 \label{fig:Arrival_OptThr}
\end{figure}
For each result in Fig. \ref{fig:Arrival_OptThr} (corresponding to a certain $\lambda$), 
the optimal threshold corresponds to the lowest average delay. For example, in Fig. \ref{MArrival Rate 45} (i.e., $\lambda=45$), the optimal threshold is 18.
As shown in Fig. \ref{MArrival Rate 30}, we can see that if the arrival rate is not high, a small enough threshold provides low delay and as the threshold increases, the delay does not change much. This is because if packets arrive at system slowly, waiting and service times of each packet in mmWave line will be probably less than service time of the sub-6 GHz server. This implies that the mmWave server does not need the aid of the sub-6 GHz server. On the contrary, adding packets to the sub-6 GHz server increases delay because the average service rate of the mmWave is much higher than that of the sub-6 GHz. 
In addition, Fig. \ref{MArrival Rate 45} to Fig. \ref{MArrival Rate 55} demonstrate that the optimal threshold decreases with the arrival rate. This is expected since a faster arrival rate may increase waiting time, which increases the chance of routing through the sub-6 GHz interface. 
\subsection{Benefits from the Sub-6 GHz with Threshold-Type Policy}
In this section, we demonstrate benefits of the sub-6 GHz interface to combat the effects of blockage and intermittent connectivity, especially under heavy traffic scenarios.  
To this end, we compare delay performance in systems with and without the sub-6 GHz. For the system with the sub-6 GHz (our integrated system), the proposed threshold-type policy is utilized. For the system without the sub-6 GHz server, no scheduling policy applies since  only mmWave interface exists in the system. To provide a more clear exhibition of our simulation results, we define relative delay improvement $\hat{W}$ as follows:
\begin{align}
\hat{W}=
\frac{\bar{W}(\text{no sub-6})-\bar{W}(\text{with sub-6})}{\bar{W}(\text{no sub-6})}	\notag;
\end{align}
where $\bar{W}(\text{with sub-6})$ and $\bar{W}(\text{no sub-6})$ denote the average delay in the integrated system and that in the system without the sub-6 GHz server, respectively. 
\begin{figure}[H]
\centering
 \subfloat[X-Z view]{
 \includegraphics[scale=.21]{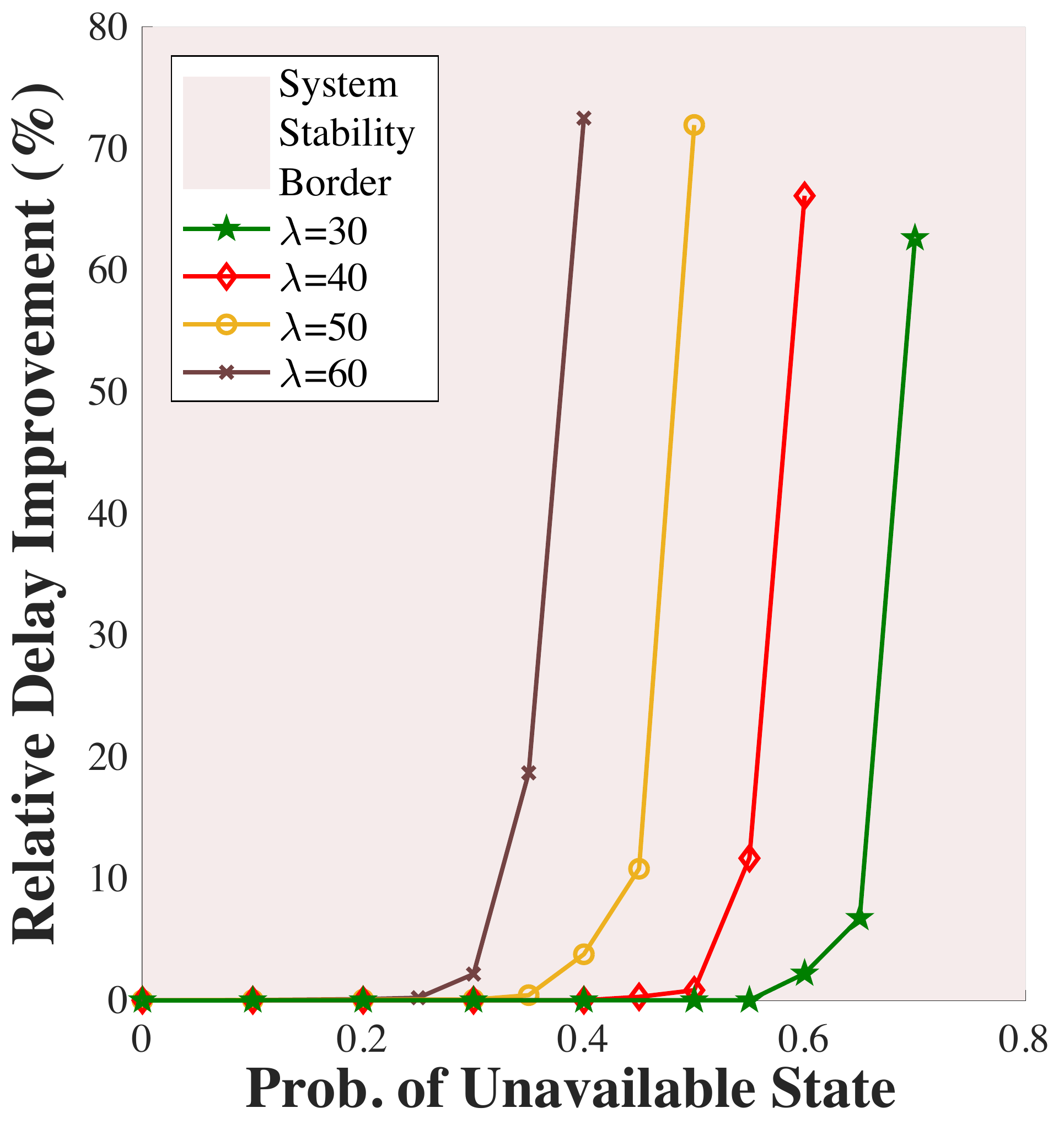}
 \label{XZview}
 }
 \subfloat[3D view]{
 \includegraphics[scale=.21]{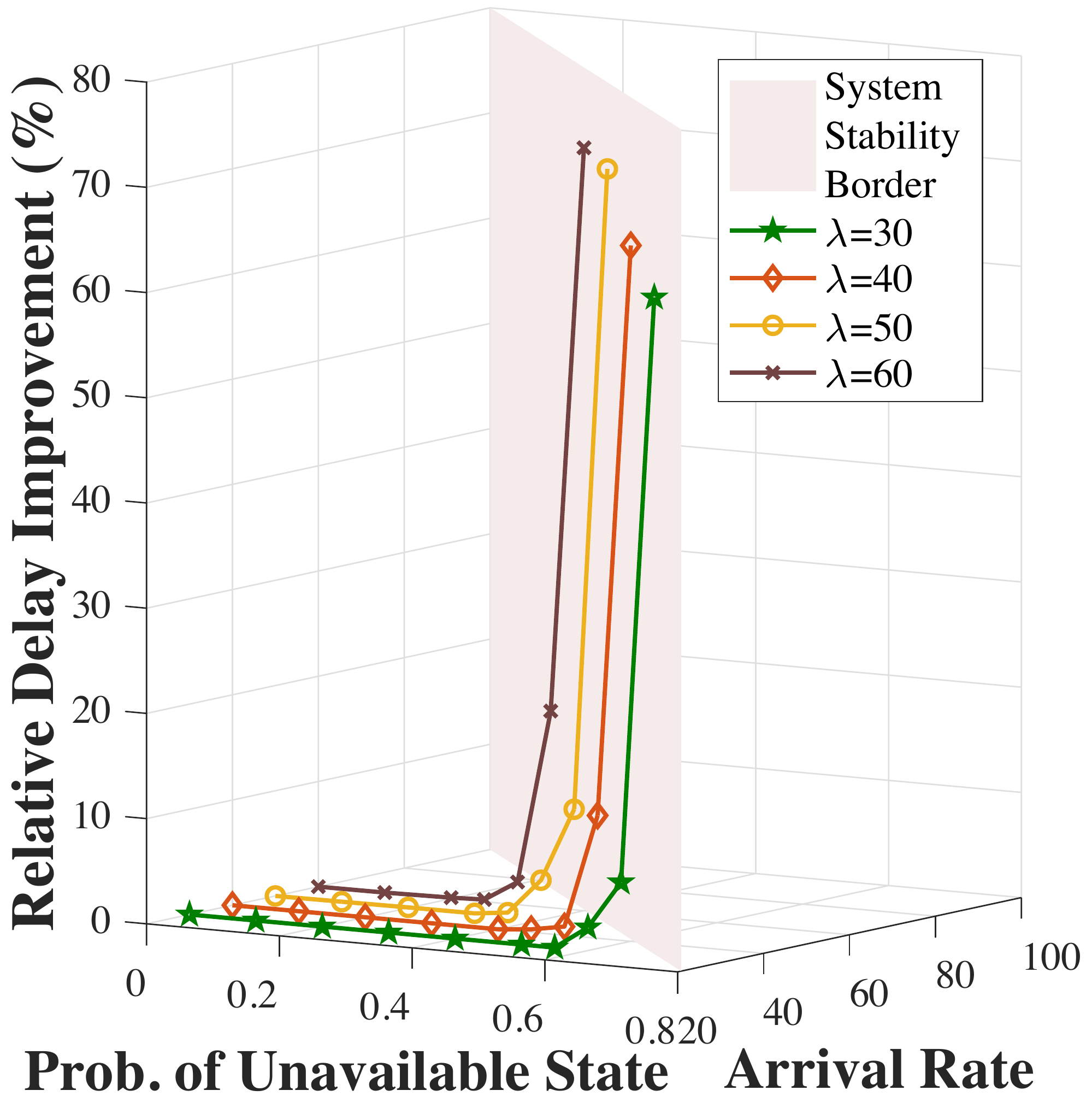}
 \label{3Dview}
 }
 \caption{Delay Performance vs Probability of Unavailable State.}
 \label{fig:delay_na_arr}
\end{figure}
In simulation, we investigate how $\hat{W}$ changes as probability of unavailable state (i.e., $p_\text{na}$) increases from $0$ to the largest value that ensures stability of the system under fixed arrival rate. We repeat the simulation for different arrival rates. From the results shown in Fig. \ref{XZview}, we observe that for a certain arrival rate, benefits of the sub-6 GHz interface becomes more pronounced as the probability of unavailable state increases. For instance, for the arrival rate of $\lambda = 60$, there is  up to $70\%$ delay reduction using the integrated architecture paired with the threshold-based policy.
Furthermore, in order to exhibit the excellent delay performance in heavy traffic scenarios, in Fig. \ref{fig:delay_na_arr} we introduce a system stability border which is a three dimensional plate that is expressed as $\lambda=\mu_\text{sub-6}+(1-p_\text{na})\mu_\text{mm}$. As shown in Fig. \ref{3Dview}, the sub-6 GHz interface becomes more beneficial as either the arrival rate or probability of unavailable state increases, i.e., heavy traffic scenarios.


\subsection{Comparison with MaxWeight Policy}

In this section, we investigate the performance of the threshold-type policy compared with the MaxWeight policy. From Fig. \ref{fig:Arrival_OptThr}, we concluded that the optimal threshold is related to the arrival rate. Hence, for each value of $\lambda$, we use the corresponding optimal threshold. From Fig. \ref{fig:comp-with-maxweight}, we note that the threshold-type policy achieves a better delay performance compared with the MaxWeight policy, while it provides a similar throughput performance. We note that the advantage of our threshold-type policy in delay performance over MaxWeight gets smaller when the arrival rate increases.
\begin{figure}[t]
\centering
 \subfloat[Delay Performance]{
 \includegraphics[scale=.34]{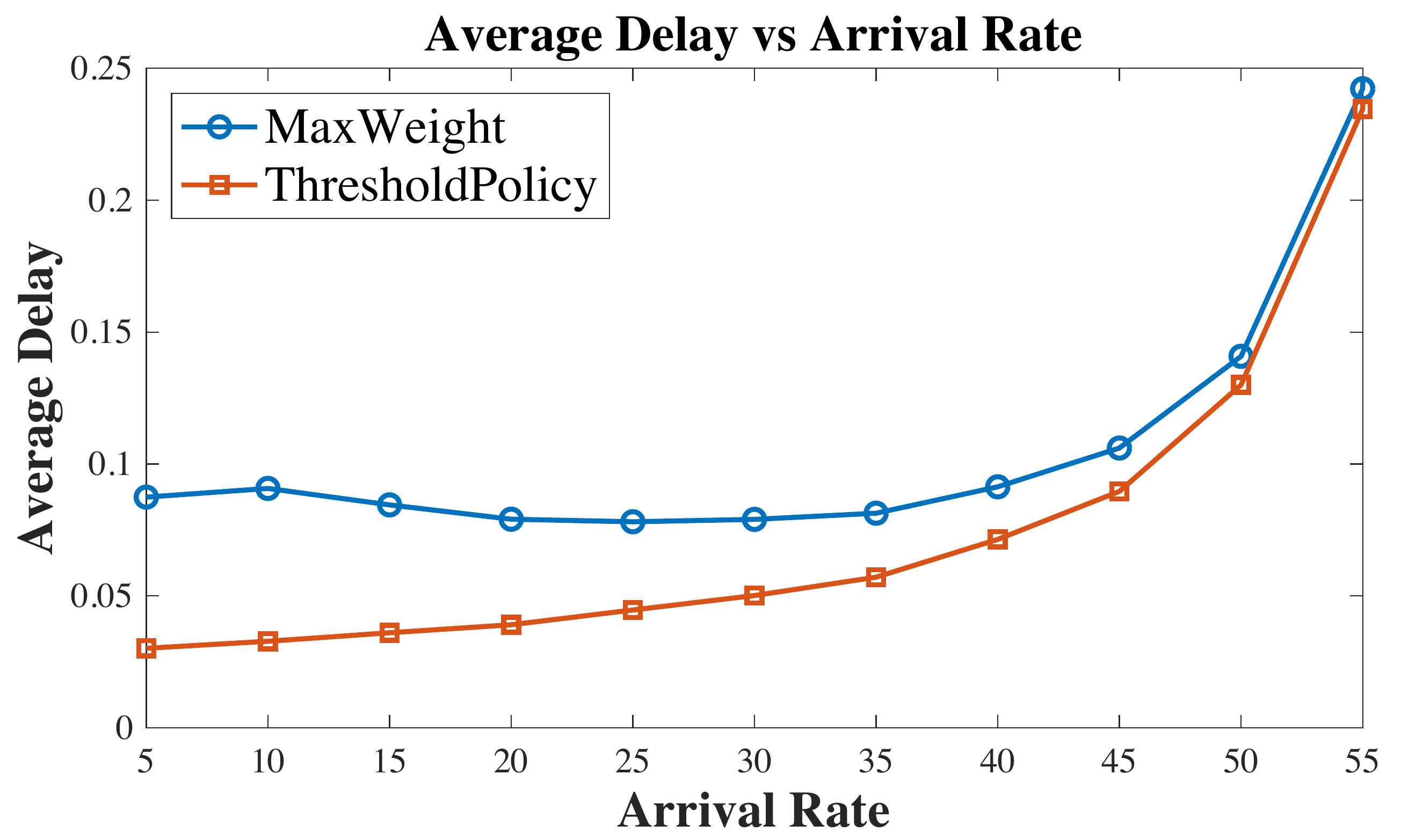}
 \label{MArrival Rate 10}
 }
 \\
 \centering
 \subfloat[Throughput Performance]{
 \includegraphics[scale=.34]{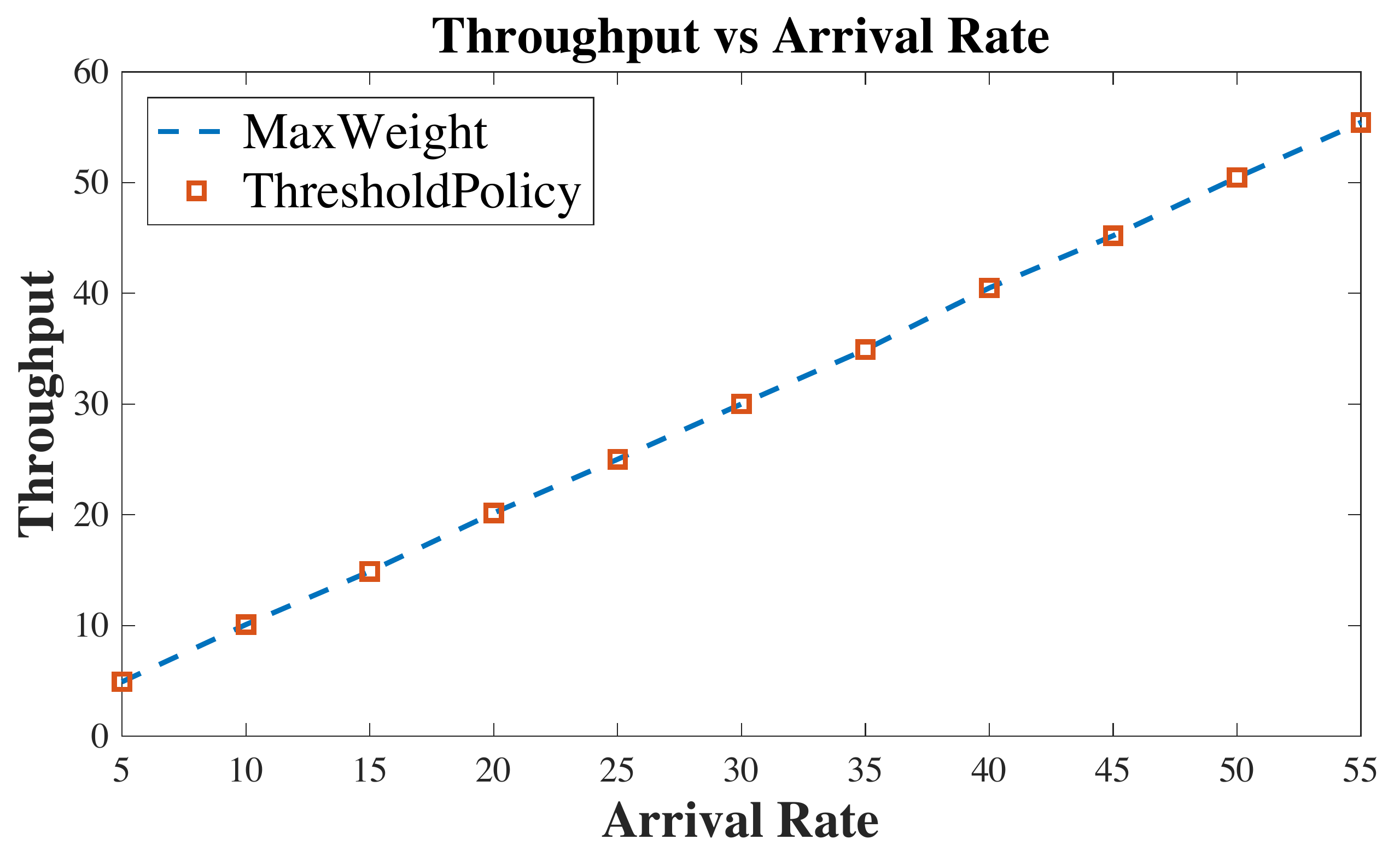}
 \label{MArrival Rate 15}
 }
 \caption{Delay and throughput performance of our proposed threshold-type policy compared with MaxWeight policy. 
 }
 \label{fig:comp-with-maxweight}
\end{figure}

\section{Conclusion}
\label{sec:conclusion}

In this paper, we considered an integrated sub-6 GHz -- mmWave architecture wherein the sub-6 GHz is used as a fallback mechanism to combat blockage and intermittent nature of the mmWave communication. In this case, the arrival packets can be transmitted using the mmWave or sub-6 GHz interface or both. We investigated the optimal packet scheduling policy such that the expected total discounted delay and the average delay are minimized and showed that the optimal policy is of the threshold-type. Through numerical results, we further investigated the delay and throughput performance of our policy to demonstrate that the threshold-type policy in fact provides a much smaller delay compared with the MaxWeight policy, while it achieves a similar throughput performance.

\appendices
\section{Proof of Theorem 1}

Note that zero function (i.e., $v=0$) satisfies all properties in Theorem \ref{thm:optimality}. Besides, it is known that for any function $f \in V$, $\lim_{n \rightarrow \infty}\mathcal{L}^{(n)}f=J_\beta$. Thus, in order to show that $J_\beta$ satisfies all properties in Theorem \ref{thm:optimality}, we start with zero function and show that $\mathcal{L}v$  satisfies the properties if $v$ satisfies the properties in Theorem \ref{thm:optimality}. 
For sake of exposition of the following proof, define $K_{\mathbf{s}}$ as the set of admissible actions in state $\mathbf{s}\in Q$.

\emph{\textbf{Property (a):}} 
By the definition of operator $\mathcal{L}$ (Eq. \eqref{eq:optimality-condition}), we show the result by respectively proving following inequalities:
\begin{align}
&\text{(1)} \ \min\limits_{u_0} v\big(u_0\left(\mathcal{A}_0\left(A_1\left(\mathbf{q}\right)\right)\right)\big)\leq \min\limits_{u_0}v\big(u_0\left(\mathcal{A}_0\left(A_{h}\left(\mathbf{q}\right)\right)\right)\big) \notag\\
&\text{(2)} \ \min\limits_{u_1} v\big(u_1\left(\mathcal{T}\left(A_1\left(\mathbf{q}\right)\right)\right)\big)\leq \min\limits_{u_1}v\big(u_1\left(\mathcal{T}\left(A_{h}\left(\mathbf{q}\right)\right)\right)\big) \notag\\
&\text{(3)} \ \min\limits_{u_2} v\big(u_2\left(\mathcal{D}_1\left(A_1\left(\mathbf{q}\right)\right)\right)\big)\leq \min\limits_{u_2}v\big(u_2\left(\mathcal{D}_1\left(A_{h}\left(\mathbf{q}\right)\right)\right)\big) \notag\\
&\text{(4)} \ \min\limits_{u_3} v\big(u_3\left(\mathcal{D}_2\left(A_1\left(\mathbf{q}\right)\right)\right)\big)\leq \min\limits_{u_3}v\big(u_3\left(\mathcal{D}_2\left(A_{h}\left(\mathbf{q}\right)\right)\right)\big).	\notag
\end{align}

\noindent \textbf{(1):} Use $\mathbf{s}_1$ and $\mathbf{s}_2$ to denote $\mathcal{A}_0\left(A_1\left(\mathbf{q}\right)\right)=\left(q_0, 1, q_1, l_2\right)$ and $\mathcal{A}_0\left(A_h\left(\mathbf{q}\right)\right)= \left(q_0+1, 0, q_1,l_2\right)$, respectively. Then, in order to show (1), we only need to show that for each $a_2 \in K_{\mathbf{s}_2}$, there exists $a_1 \in K_{\mathbf{s}_1}$ such that $v(a_2(\mathbf{s}_2))\geq v(a_1(\mathbf{s}_1))$. Same logic will also be used in the following proof for other case s.

Generally, $\{A_{h}, A_1\}\subseteq K_{\mathbf{s}_2}$ and $A_h \in K_{\mathbf{s}_1}$. Then, we have 
\begin{align} 
v\left(A_{h}\left(\mathbf{s}_2\right)\right) &\stackrel{\text{(a)}}{\geq} v\left(A_{1}\left(\mathbf{s}_2\right)\right)=  v\left(A_h\left(\mathbf{s}_1\right)\right). \nonumber
\end{align}
If $l_2=0$, then $A_2\in K_{\mathbf{s}_2}$ and $A_r\in K_{\mathbf{s}_1}$. Notice that $A_{r}\left(\mathbf{s}_1\right)=\left(q_0, 0, q_1, 1\right)$ by property (c). In the case, we obtain
$$v\left(A_{2}\left(\mathbf{s}_2\right)\right)=v\left(A_{r}\left(\mathbf{s}_1\right)\right).$$
If $l_2=0$ and $q_1\geq 1$, then $A_r\in A_{\mathbf{s}_2}$ and we have
\begin{align}
v\left(A_r\left(\mathbf{s}_2\right)\right)&\stackrel{\text{(b)}}{\geq} v\left(A_{2}\left(\mathbf{s}_2\right)\right)=v\left(A_{r}\left(\mathbf{s}_1\right)\right).	\notag
\end{align}


\noindent \textbf{(2):} 
Denote $\mathcal{T}\left(A_1\left(\mathbf{q}\right)\right)=\left(q_0-1, 0, q_1+1, l_2\right)$ as $\mathbf{s}_3$ and denote $\mathcal{T}\left(A_h\left(\mathbf{q}\right)\right)=\left(q_0, 0, q_1, l_2\right)$ as $\mathbf{s}_4$.
Generally, $A_h \in K_{\mathbf{s}_3}$ and $\{A_{h}, A_1\}\subseteq K_{\mathbf{s}_4}$. Thus, we have
\begin{align} 
v\left(A_{h}\left(\mathbf{s}_4\right)\right) &\stackrel{\text{(a)}}{\geq} v\left(A_{1}\left(\mathbf{s}_4\right)\right)\notag\\
&\stackrel{\text{(c)}}{\geq}v\left(\mathcal{T}\left(q_0-1, 1, q_1,l_2\right)\right)= v\left(A_h\left(\mathbf{s}_3\right)\right). \nonumber
\end{align}
If $l_2=0$, then $A_2 \in K_{\mathbf{s}_4}$ and $A_r \in K_{\mathbf{s}_3}$. Then, we obtain $$v\left(A_{2}\left(\mathbf{s}_4\right)\right)=v\left(A_{r}\left(\mathbf{s}_3\right)\right).$$
If $l_2=0$ and $q_1\geq 1$, then $A_r\in K_{\mathbf{s}_4}$ and we have
\begin{align}
	v(A_{r}(\mathbf{s}_4))&\stackrel{\text{(b)}}{\geq} v(A_{2}(\mathbf{s}_4))=v(A_{r}(\mathbf{s}_3)).\notag
\end{align}
\noindent \textbf{(3):} 
Denote $\mathcal{D}_1\left(A_1\left(\mathbf{q}\right)\right)=\left(q_0-1, 1, \left(q_1-1\right)^+, l_2\right)$ as $\mathbf{s}_5$ and denote $\mathcal{D}_1\left(A_h\left(\mathbf{q}\right)\right)=\left(q_0, 0, \left(q_1-1\right)^+, l_2\right)$ as $\mathbf{s}_6$. Generally, $\{A_h,A_1\}\subseteq K_{\mathbf{s}_6}$ and $A_h \in K_{\mathbf{s}_5}$.
Then, we get 
\begin{align} 
v\left(A_{h}\left(\mathbf{s}_6\right)\right) &\stackrel{\text{(a)}}{\geq} v\left(A_{1}\left(\mathbf{s}_6\right)\right)= v\left(A_{h}\left(\mathbf{s}_5\right)\right). \nonumber
\end{align}
If $l_2=0$, then $A_2 \in K_{\mathbf{s}_6}$ and $A_r \in K_{\mathbf{s}_5}$. Then, we have
 $$v\left(A_{2}\left(\mathbf{s}_6\right)\right)=v\left(A_{r}\left(\mathbf{s}_5\right)\right).$$
If $l_2=0$ and $\left(q_1-1\right)^+\geq 1$, then $A_r \in K_{\mathbf{s}_6}$. Thus, we get
\begin{align}
v\left(A_{r}\left(\mathbf{s}_6\right)\right)&\stackrel{\text{(b)}}{\geq} v\left(A_{2}\left(\mathbf{s}_6\right)\right)=v\left(A_{r}\left(\mathbf{s}_5\right)\right).	\notag
\end{align}


\noindent \textbf{(4):} 
Denote $\mathcal{D}_2\left(A_1\left(\mathbf{q}\right)\right)=\left(q_0-1, 1, q_1, 0\right)$ as $\mathbf{s}_7$ and denote $\mathcal{D}_2\left(A_h\left(\mathbf{q}\right)\right)=\left(q_0, 0, q_1, 0\right)$ as $\mathbf{s}_8$. 
Generally, $\{A_{h}, A_1, A_2\} \subseteq K_{\mathbf{s}_8}$ and $\{A_{h}, A_r\} \subseteq K_{\mathbf{s}_7}$. Then, we obtain
\begin{align}
&v\left(A_{h}\left(\mathbf{s}_8\right)\right) \stackrel{\text{(a)}}{\geq} v\left(A_{1}\left(\mathbf{s}_8\right)\right) =  v\left(A_h\left(\mathbf{s}_7\right)\right),\notag\\
&v\left(A_{2}\left(\mathbf{s}_8\right)\right)=v\left(A_{r}\left(\mathbf{s}_7\right)\right).	\notag
\end{align}
If $q_1\geq 1$, $A_r \in K_{\mathbf{s}_8}$, then we have
\begin{align}
	v\left(A_{r}\left(\mathbf{s}_8\right)\right)&\stackrel{\text{(b)}}{\geq} v\left(A_{2}\left(\mathbf{s}_8\right)\right)=v\left(A_{r}\left(\mathbf{s}_7\right)\right).\notag
\end{align}

Therefore we conclude that $\mathcal{L}v\left(A_1\left(\mathbf{q}\right)\right)\leq \mathcal{L}v\left(A_{h}\left(\mathbf{q}\right)\right)$.

\emph{\textbf{Property (b):}}
Notice that by property (c), $A_r(\mathbf{q})=(q_0, 0, q_1, 1)$, if $l_1=1$; otherwise, $A_r(\mathbf{q})=(q_0, 0, q_1-1, 1)$.

\emph{As for the case $l_1=0$:} the only thing for us to note is that if $q_1=1$, then $\mathcal{D}_1\left(A_2\left(\mathbf{q}\right)\right)=\left(q_0-1,0,0,1\right)$ and $\mathcal{D}_1\left(A_r\left(\mathbf{q}\right)\right)=\left(q_0,0,0,1\right)$. Thus, by monotonicity property (e), we directly obtain that $\min_{u_2}{v\left(u_2\left(\mathcal{D}_1\left(A_2\left(\mathbf{q}\right)\right)\right)\right)}\leq \min_{u_2}{v\left(u_2\left(\mathcal{D}_1\left(A_r\left(\mathbf{q}\right)\right)\right)\right)}$. With this and similar argument in  proof of property (a), we can obtain that for $q_0\geq 1$ and $q_1\geq 1$, $\mathcal{L}v\left(A_2\left(q_0, 0, q_1,0\right)\right)\leq \mathcal{L}v\left(A_r\left(q_0, 0, q_1,0\right)\right)$.

\emph{As for the case $l_1=1$:} the case can be proved with similar argument and logic in proof for property (a).

\emph{\textbf{Property (c):}}
The property can be easily shown with similar logic and argument in proof for property (a) and (b).

\emph{\textbf{Property (d):}}
With similar argument, we easily obtain that $\min_{u_0} v\left(u_0\left(\mathcal{A}_0\left(A_1\left(\mathbf{q}\right)\right)\right)\right)\leq \min_{u_0}v\left(u_0\left(\mathcal{A}_0\left(A_2\left(\mathbf{q}\right)\right)\right)\right)$, where $\mathbf{q}=\left(q_0,0,0,0\right)$. 
It remains to show 
\begin{align}
&\mu_\text{p}B_1+p_\text{a}\mu_\text{mm}B_2+\mu_\text{sub-6}B_2\notag\\
\leq &\mu_\text{p}B_3+p_\text{a}\mu_\text{mm}B_3+\mu_\text{sub-6}B_4.	
\label{*}
\end{align}
where 
\begin{align}
& B_1\triangleq\min_{u_1} v\left(u_1\left(q_0-1, 0, 1, 0\right)\right)\notag\\
& B_2\triangleq\min_{u_2} v\left(u_2\left(q_0-1, 1, 0, 0\right)\right)=\min_{u_3} v\left(u_3\left(q_0-1, 1, 0, 0\right)\right)\notag\\
& B_3\triangleq \min_{u_1}v\left(u_1\left(q_0-1, 0, 0, 1\right)\right)=\min_{u_2}v\left(u_2\left(q_0-1, 0, 0, 1\right)\right)\notag\\
& B_4\triangleq \min_{u_3} v\left(u_3\left(q_0-1, 0, 0, 0\right)\right)\notag. 	
\end{align}
Note that 
$B_4 \leq B_1 \leq B_2 \leq B_3 \label{**}
$, which can be shown via similar argument in previous proof. 




\begin{enumerate}
	\item \emph{Case that $q_0>1$:} notice that by property (a) and (d), $B_4=v\left(q_0-2,1,0,0\right)$ and by property (a), $B_3=v\left(q_0-2,1,0,1\right)$. 
\begin{enumerate}
\item  If $v\left(q_0-1,1,0,0\right)<v\left(q_0-2,1,0,1\right)$, then $B_2=v(q_0-1,1,0,0)<B_3$. It implies that the sub-6 GHz interface is not preferable in the case. 
Then, $B_2-B_4\leq \frac{1}{\mu_\text{sub-6}}$
Besides, since the least time that is required for an event to happen is $\min\left\{\frac{1}{\lambda},\frac{1}{\mu_\text{p}},\frac{1}{p_\text{a}\mu_\text{mm}},\frac{1}{\mu_\text{sub-6}}\right\}$=$\min\left\{\frac{1}{\lambda},\frac{1}{\mu_\text{p}},\frac{1}{p_\text{a}\mu_\text{mm}}\right\}$, we have $B_3-B_2\geq \min\left\{\frac{1}{\lambda},\frac{1}{\mu_\text{p}},\frac{1}{p_\text{a}\mu_\text{mm}}\right\}$. 
Therefore, 
we have 
\begin{align}
&\mu_\text{p}\left(B_1-B_3\right)+p_\text{a}\mu_\text{mm}\left(B_2-B_3\right)+\mu_\text{sub-6}\left(B_2-B_4\right)\notag\\ 
\leq & \left(\mu_\text{p}+p_\text{a}\mu_\text{mm}\right)\left(B_2-B_3\right)+\mu_\text{sub-6}\left(B_2-B_4\right)\notag\\
\leq & \frac{\mu_\text{sub-6}}{\mu_\text{sub-6}}-\left(\mu_\text{p}+p_\text{a}\mu_\text{mm}\right)\min\left\{\frac{1}{\lambda},\frac{1}{\mu_\text{p}},\frac{1}{p_\text{a}\mu_\text{mm}}\right\}
\leq  0.	\notag
\end{align}

\item  If $v\left(q_0-1,1,0,0\right)\geq v\left(q_0-2,1,0,1\right)$, which implies the sub-6 GHz interface should be utilized in the case, then $B_2=B_3=v\left(q_0-2,1,0,1\right)$ and $B_1=v\left(q_0-2,0,1,1\right)$. For the best case, the smallest difference between $B_2$ and $B_1$ is $\frac{1}{\mu_\text{p}}$.
Besides, $B_3-B_4\leq  \frac{1}{\mu_\text{sub-6}}$. Thus, inequality \eqref{*} holds.
\end{enumerate}
\item\emph{Case that $q_0=1$:} we have $B_2-B_4\leq \frac{1}{\mu_\text{p}}+\frac{1}{p_\text{a}\mu_\text{mm}}$ and $B_3-B_2\geq \min\left\{\frac{1}{\lambda},\frac{1}{\mu_\text{p}},\frac{1}{p_\text{a}\mu_\text{mm}}\right\}$. Then, the remaining proof is the same as that in (1) for case $q_0>1$.
\end{enumerate}

\emph{\textbf{Property (e):} Now we check monotonicity:}
\begin{align}
&\mathcal{L}v\left(q_0+1,l_1,q_1,l_2\right) \notag\\
=& \left(q_0+l_1+q_1+l_2+1\right)\notag\\&+\beta\min_{\mathbf{u}}\big\{\lambda v\left(u_0\left(\mathcal{A}_0\left(q_0+1,l_1,q_1,l_2\right)\right)\right)\notag\\
&+\mu_\text{p}v\left(u_1\left(\mathcal{T}\left(q_0+1,l_1,q_1,l_2\right)\right)\right)
\notag\\
&+ p_\text{a}\mu_\text{mm}v\left(u_2\left(\mathcal{D}_1\left(q_0+1,l_1,q_1,l_2\right)\right)\right)\notag\\
&+\mu_\text{sub-6}v\left(u_3\left(\mathcal{D}_2\left(q_0+1,l_1,q_1,l_2\right)\right)\right)\big\}\notag\\
\stackrel{\text{(e)}}{\geq} & \left(q_0+l_1+q_1+l_2\right)+\beta\min_{\mathbf{u}}\big \{\lambda v\left(u_0\left(\mathcal{A}_0\left(q_0,l_1,q_1,l_2\right)\right)\right)\notag\\&+\mu_\text{p}v\left(u_1\left(\mathcal{T}\left(q_0,l_1,q_1,l_2\right)\right)\right)
\notag\\
&+p_\text{a}\mu_\text{mm}v\left(u_2\left(\mathcal{D}_1\left(q_0,l_1,q_1,l_2\right)\right)\right)\notag\\
&+\mu_\text{sub-6}v\left(u_3\left(\mathcal{D}_2\left(q_0,l_1,q_1,l_2\right)\right)\right)\big\} \notag \\
=&\mathcal{L}v\left(q_0,l_1,q_1,l_2\right).	\notag
\end{align}
Similarly, we can  obtain that $\mathcal{L}v\left(q_0,l_1,q_1+1,l_2\right)\geq \mathcal{L}v\left(q_0,l_1,q_1,l_2\right)$,
$\mathcal{L}v\left(q_0,1,q_1,l_2\right)\geq \mathcal{L}v\left(q_0,0,q_1,l_2\right)$ and  $\mathcal{L}v\left(q_0,l_1,q_1,1\right)\geq \mathcal{L}v\left(q_0,l_1,q_1,0\right)$. 

\section{Proof of Lemma 1}
Note that $J^0_\beta\left(x,q_1,l_2\right)=x+q_1+l_2$ and $J^0_\beta \in \mathcal{F}$ obviously. By Eq. \eqref{valueiteration}, it remains to show that $T^n_\beta \in \mathcal{F}$ and then $J^{n+1}_\beta \in \mathcal{F}$ given $J^n_\beta \in \mathcal{F}$. 
Before our proof, we provide some properties extended from Definition \ref{F_class}, which will be used in the following proof.

\emph{Extended properties from Definition \ref{F_class}:}
\begin{align}
	&2f\left(x,q_1,1\right)\leq  f\left(x+1,q_1,1\right)+f\left(x-1,q_1,1\right) \label{extension1}\\
	&2f\left(0,q_1,1\right)\leq f\left(1,q_1,1\right)+f\left(0,q_1-1,1\right) \label{extension2}\\
	&2f\left(0,q_1,1\right)\leq f\left(0,q_1+1,1\right)+f\left(0,q_1-1,1\right) \label{extension4}\\
	&2f\left(x+1,q_1,0\right)\leq f\left(x+2,q_1,0\right)+f\left(x,q_1,0\right)\label{extension6}\\
	&2f\left(0,q_1+1,0\right)\leq f\left(0,q_1,0\right)+f\left(0, q_1+2,0\right)\label{extension8}\\
	&f\left(x,q_1,1\right)+f\left(x-1,q_1+1,1\right)\notag\\
	&\ \ \ \ \ \ \ \ \ \ \ \ \ \ \leq   f\left(x,q_1+1,1\right)+f\left(x-1,q_1,1\right) \label{extension3}\\
		&f\left(0,q_1+1,0\right)+f\left(0,q_1+1,1\right)\notag\\
	&\ \ \ \ \ \ \ \ \ \ \ \ \ \ \leq   f\left(0,q_1,1\right)+f\left(1,q_1+1,0\right)\label{extension5}	\\
			&f\left(x+1,q_1,0\right)+f\left(x, q_1+1,0\right)\notag\\
	&\ \ \ \ \ \ \ \ \ \ \ \ \ \  \leq  f\left(x+1,q_1+1,0\right)+f\left(x, q_1,0\right)\label{extension7}
\end{align}
These properties can be obtained from combinations of certain equations in Definition \ref{F_class}. For lack of space, we take Eq. \eqref{extension1} for example, it is obtained by adding Eq. \eqref{A2_x} with $x$ replaced by $x-1$ and Eq. \eqref{supermodular1}.

Given $J_\beta^{n}\in \mathcal{F}$, we first show that $T^n_\beta \in \mathcal{F}$.
\begin{enumerate}
\item\emph{For Eq. (6):} If $T^n_\beta\left(x+2,q_1,0\right)=J^n_\beta\left(x+2,q_1,0\right)$, then
\begin{align}	T^n_\beta\left(x+1,q_1,0\right)+&T^n_\beta\left(x+1,q_1,1\right)\notag\\
	\stackrel{\text{Def. \ref{intermediate value}}}{\leq} & J^n_\beta\left(x+1,q_1,0\right)+J^n_\beta\left(x+1,q_1,1\right)\notag\\
	\stackrel{\text{\eqref{A2_x}}}{\leq}&J^n_\beta\left(x,q_1,1\right)+J^n_\beta\left(x+2,q_1,0\right).\notag
\end{align}
If $T^n_\beta\left(x+2,q_1,0\right)=J^n_\beta\left(x+1,q_1,1\right)$, then
\begin{align}
T^n_\beta\left(x+1,q_1,0\right)+&T^n_\beta\left(x+1,q_1,1\right)\notag\\
	\stackrel{\text{Def. \ref{intermediate value}}}{\leq} &
	 J^n_\beta\left(x,q_1,1\right)+J^n_\beta\left(x+1,q_1,1\right).\notag
\end{align}
Similarly, we can show that Eq. \eqref{A2_q1},   \eqref{Ar} and \eqref{switch} hold.







\item \emph{For Eq. (10):} If $T^n_\beta\left(x,q_1,0\right)=J^n_\beta\left(x,q_1,0\right)$, then
\begin{align}
	T^n_\beta\left(x,q_1,1\right)+&T^n_\beta\left(x+1,q_1,0\right)\notag\\
	\stackrel{\text{Def. \ref{intermediate value}}}{\leq}& J^n_\beta\left(x,q_1,1\right)+J^n_\beta\left(x+1,q_1,0\right)\notag\\
	\stackrel{\eqref{supermodular1}}{\leq}&J^n_\beta\left(x,q_1,0\right)+J^n_\beta\left(x+1,q_1,1\right).\notag
\end{align}
If $x\geq 1$ and $T^n_\beta\left(x,q_1,0\right)=J^n_\beta\left(x-1,q_1,1\right)$, then
\begin{flalign}
	&T^n_\beta\left(x,q_1,1\right)+T^n_\beta\left(x+1,q_1,0\right)\notag \\
	&\stackrel{\text{Def. \ref{intermediate value}}}{\leq} 2J^n_\beta\left(x,q_1,1\right)\stackrel{\eqref{extension1}}{\leq}J^n_\beta\left(x-1,q_1,1\right)+J^n_\beta\left(x+1,q_1,1\right).\notag 
\end{flalign}
If $x=0$, $q_1\geq 1$ and $T^n_\beta\left(0,q_1,0\right)=J^n_\beta\left(0,q_1-1,1\right)$, then 
\begin{align}
	&T^n_\beta\left(0,q_1,1\right)+T^n_\beta\left(1,q_1,0\right)\notag\\
	&\stackrel{\text{Def. \ref{intermediate value}}}{\leq} 2J^n_\beta\left(0,q_1,1\right)\stackrel{\text{\eqref{extension2}}}{\leq}J^n_\beta\left(0,q_1-1,1\right)+J^n_\beta\left(1,q_1,1\right).\notag
\end{align}

\item \emph{For Eq. (11):} If $T^n_\beta\left(x,q_1,0\right)=J^n_\beta\left(x,q_1,0\right)$, then
\begin{align}
	T^n_\beta\left(x,q_1,1\right)+&T^n_\beta\left(x,q_1+1,0\right)\notag\\
	\stackrel{\text{Def. \ref{intermediate value}}}{\leq}& J^n_\beta\left(x,q_1,1\right)+J^n_\beta\left(x,q_1+1,0\right)\notag\\
	\stackrel{\text{\eqref{supermodular2}}}{\leq}&J^n_\beta\left(x,q_1,0\right)+J^n_\beta\left(x,q_1+1,1\right).\notag
\end{align}

\noindent If $x\geq 1$ and $T^n_\beta\left(x,q_1,0\right)=J^n_\beta\left(x-1,q_1,1\right)$, then
\begin{align}
	T^n_\beta\left(x,q_1,1\right)+&T^n_\beta\left(x,q_1+1,0\right)\notag\\
		\stackrel{\text{Def. \ref{intermediate value}}}{\leq}& J^n_\beta\left(x,q_1,1\right)+J^n_\beta\left(x-1,q_1+1,1\right)\notag\\
	\stackrel{\text{\eqref{extension3}}}{\leq}&J^n_\beta\left(x-1,q_1,1\right)+J^n_\beta\left(x,q_1+1,1\right).\notag
\end{align}
If $x=0$, $q_1\geq 1$ and $T^n_\beta\left(0,q_1,0\right)=J^n_\beta\left(0,q_1-1,1\right)$, then
\begin{align}
	&T^n_\beta\left(0,q_1,1\right)+T^n_\beta\left(0,q_1+1,0\right)\notag\\
	&\stackrel{\text{Def. \ref{intermediate value}}}{\leq}2J^n_\beta\left(0,q_1,1\right)\stackrel{\eqref{extension4}}{\leq}J^n_\beta\left(0,q_1-1,1\right)+J^n_\beta\left(0,q_1+1,1\right).\notag
\end{align}




\item \emph{For Eq. (12):} If $T^n_\beta\left(x+1,q_1,l_2\right)=J^n_\beta\left(x+1,q_1,l_2\right)$, then
\begin{align}
	T^n_\beta\left(x,q_1,l_2\right)\stackrel{\text{Def. \ref{intermediate value}}}{\leq}J^n_\beta\left(x,q_1,l_2\right)\stackrel{\eqref{mono1}}{\leq}J^n_\beta\left(x+1,q_1,l_2\right).\notag
\end{align}
If $l_2=0$ and $T^n_\beta\left(x+1,q_1,0\right)=J^n_\beta\left(x,q_1,1\right)$, then
\begin{align}
	T^n_\beta\left(x,q_1,0\right)\stackrel{\text{Def. \ref{intermediate value}}}{\leq} J^n_\beta\left(x,q_1,0\right)\stackrel{\eqref{mono3}}{\leq}J^n_\beta\left(x,q_1,1\right).\notag
\end{align}
Similarly, we obtain Eq. \eqref{mono2} and Eq. \eqref{mono3}. 
\end{enumerate}
Next, we show that $J^{n+1}_\beta \in \mathcal{F}$. According to Eq. \eqref{valueiteration}, we show four terms, say $\lambda$, $\mu_\text{p}$, $\mu_\text{mm}$ and $\mu_\text{sub-6}$ terms, satisfy properties in Definition \ref{F_class}, respectively.
\begin{enumerate}
	\item \emph{For Eq. (6):} the difficulty falls in the $\mu_\text{p}$ and $\mu_\text{sub-6}$ terms. For the $\mu_\text{p}$ term, the difficulty falls in the case with $x=0$, which can be proved with Eq. \eqref{extension5}.
For the $\mu_\text{sub-6}$ term, Eq. \eqref{A2_x} reduces to Eq. \eqref{extension6}.

\item\emph{For Eq. (7):} the $\lambda$ term obviously holds. 
As for the $\mu_\text{p}$ term, the difficulty falls in the case with $x=0$. Actually, it reduces to Eq. \eqref{Ar}.
As for the $\mu_\text{mm}$ term, the difficulty falls in the case with $q_1=0$. In the case, Eq. \eqref{A2_q1} reduces to equality.
As for $\mu_\text{sub-6}$ term, Eq. \eqref{A2_q1} reduces to Eq. \eqref{extension7}.

\item \emph{For Eq. (8):} the $\lambda$ and $\mu_\text{p}$ terms obviously hold.
As for the $\mu_\text{mm}$ term, the difficulty falls in the case with $q_1=0$, where $T^n_\beta(0,0,0)\leq T^n_\beta(0,1,0)$. In fact, the inequality holds by Eq. \eqref{mono2}.
As for the $\mu_\text{sub-6}$ term, Eq. \eqref{Ar} reduces to Eq. \eqref{extension8}.

\item \emph{For Eq. (9):} the $\lambda$, $\mu_\text{p}$ with $x\geq 1$, $\mu_\text{mm}$ with $q_1\geq 1$ and $\mu_\text{sub-6}$ terms hold obviously.
As for the $\mu_\text{p}$ with $x=0$ term, Eq. \eqref{switch} reduces to an equation.
As for the $\mu_\text{mm}$ with $q_1=0$ term, Eq. \eqref{switch} reduces to Eq. \eqref{mono1} with $q_1=0$.

\item \emph{For Eq. (10):} it is obvious that the $\lambda$, $\mu_\text{p}$ with $x\geq 1$, and $\mu_\text{mm}$ terms hold.
Notice that as for the $\mu_\text{sub-6}$ term, Eq. \eqref{supermodular1} reduces to an equality. As for the $\mu_\text{p}$ term with $x=0$, Eq. \eqref{supermodular1} reduces to Eq. \eqref{supermodular2}.

\item \emph{For Eq. (11):} the difficulty falls in
the $\mu_\text{sub-6}$ and $\mu_\text{mm}$ with $q_1=0$ terms. For both of the cases, Eq. \eqref{supermodular2} reduces to an equality.

\item \emph{For Eq. (12):} the only difficulty falls in the $\mu_\text{p}$ term with $x=0$, in which case, Eq. \eqref{mono1} reduces to Eq. \eqref{mono2} with $x=0$.

\item \emph{For Eq. (13):} the only difficulty falls in the $\mu_\text{mm}$ term with $q_1=0$. In the case, Eq. \eqref{mono2} reduces to an equality.

\item \emph{For Eq. (14):} the only difficulty falls in the $\mu_\text{sub-6}$ term. Actually, in the case, Eq. \eqref{mono3} reduces to an equality.
\end{enumerate}
\section{Proof of Lemma 2}
According to Lemma \ref{thm:F_class}, for each $n\in \mathbb{N}$, $J^n_\beta$ satisfies properties \eqref{A2_x}, \eqref{A2_q1}, \eqref{Ar} and \eqref{extension5}. It implies that for either the case $x>0$ or $x=0$, $J^n_\beta\left(x+1,q_1,0\right)-J^n_\beta\left(x,q_1,1\right)$ or $J^n_\beta\left(0,q_1+1,0\right)-J^n_\beta\left(0,q_1,1\right)$ increases as $x+q_1$ increases (due to increase of $x$ or $q_1$ or both). 
In other words, the difference between costs resulted from not-adding-to-sub-6 and adding-to-sub-6 increases as the number of packets in FastLane increases. It is known that $J^n_\beta\left(0,1,0\right)\leq J^n_\beta\left(1,0,0\right)\leq J^n_\beta\left(0,0,1\right)$, which means that it's better to hold the packet in FastLane when there is only one packet in the system. As $x+q_1$ increases, the difference becomes positive, which means that adding-to-sub-6 obtains priority. To sum up, there exists a certain threshold for the queue length of FastLane above which we should add a packet to the sub-6 GHz interface. 
\bibliographystyle{unsrt}
\bibliography{References}
 
\end{document}